\documentclass[11pt,a4paper]{article}

\usepackage[left=2.7cm,right=2.7cm,top=2.7cm,bottom=3.4cm]{geometry}
\usepackage[utf8]{inputenc}
\usepackage[T1]{fontenc}
\usepackage[round,authoryear]{natbib}
\usepackage[english]{babel}
\usepackage{mathtools,amsfonts,amssymb,amsthm,mathrsfs}
\usepackage{cases}
\usepackage{tabvar}
\usepackage{enumerate}
\usepackage{tabularx}
\usepackage{graphicx}
\usepackage{url}
\usepackage{numprint}
\usepackage{titlesec}
\usepackage{multirow}
\usepackage{authblk}
\usepackage[colorlinks=true,linkcolor=darkblue,citecolor=darkblue,urlcolor=darkblue]{hyperref}

\definecolor{darkblue}{rgb}{0,0,0.5}

\newcolumntype{Y}{>{\centering\arraybackslash}X}

\newtheorem{proposition}{Proposition}[section]

\theoremstyle{definition}

\newtheorem{remark}{Remark}[section]

\newcommand{\reels}{\mathbb{R}}

\newcommand{\proba}{P}
\newcommand{\argmax}{\operatorname*{arg \; max}}
\DeclareMathOperator{\dd}{\mbox{}}
\newcommand{\dt}{\dd{}\mathrm{d}t}

\title{\textbf{Theoretical analysis and improvements in cubic transmutations of probability distributions}}
\author[a]{Issa Cherif Geraldo}
\author[b]{Edoh Katchekpele}
\author[b]{Tchilabalo Abozou Kpanzou}
\affil[a]{\footnotesize Laboratoire d'Analyse, de Mod\'elisations Math\'ematiques et Applications (LAMMA), D\'epartement de Math\'ematiques, Facult\'e des Sciences, Université de Lomé, 1 B.P. 1515 Lom\'e 1, Togo.}
\affil[b]{\footnotesize Laboratoire de Mod\'elisation Math\'ematique et d'Analyse Statistique Décisionnelle (LaMMASD), D\'epartement de Math\'ematiques, Facult\'e des Sciences et Techniques, Universit\'e de Kara, Kara, Togo.}

\date{}

\begin{document}
	
\renewcommand{\proofname}{\textbf{\normalshape Proof}}

\vspace{-2cm}
\maketitle

\begin{abstract}
	In statistics, processed data are becoming increasingly complex, and classical probability distributions are limited in their ability to model them. This is why, to better model data, extensive work has been conducted on extending classical probability distributions. Generally, this extension is achieved by transforming the cumulative distribution function of a baseline distribution through the addition of one or more parameters to enhance its flexibility. Cubic transmutation (CT) is one of the most popular methods for such extensions. However, CT does not have a unique definition because different approaches for CT have been proposed in the literature but are yet to be compared. The main goal of this paper is to compare these different approaches from both theoretical and empirical viewpoints. We study the relationships between the different approaches and we propose modified versions based on the extension of parameter ranges. The results are illustrated using Pareto distribution as baseline distribution. \\
	 
	
	\noindent \textbf{AMS Subject Classification (MSC2020)} -- 60E05, 62E15, 62F10, 62F99, 62P99.
\end{abstract}

\section{Introduction}

In statistics, the data to be processed are becoming increasingly complex, and classical probability distributions are limited in their ability to model them. This is why, in order to better model the data, considerable work has focused on extending classical probability distributions. Generally, this extension is achieved by transforming a given baseline cumulative distribution function (cdf) into another by adding one or more new parameters to enhance its flexibility. This is how, over the past few decades, scientific research has seen the emergence of new families of probability distributions. The reader can consult references such as \citet{ahmad-et-al-2019, ahmad-et-al-2022} and \citet{tahir-cordeiro-2016}.

Among the techniques for extending probability distributions, transmutation is one of the most used. The first form of transmutation is the quadratic transmutation introduced by \cite{shaw-buckley-2009}. It consists of transforming a baseline cdf $G(x)$ (that can depend on one or more parameters) in the form of a new cdf defined for all $x \in \reels$ by
\begin{equation}
	\label{eq:quad-trans}
	F_S(x) = (1+\lambda) G(x) - \lambda G^2(x),
\end{equation}
where $\lambda$ is an additional parameter such that $\vert \lambda \vert \leqslant 1$. It can be noticed that the value $\lambda = 0$ enables to get the baseline cdf $G(x)$. Several authors (see \cite{tahir-cordeiro-2016}, \cite{dey-et-al-2021}, \cite{imliyangba-et-al-2021}) have applied the formula \eqref{eq:quad-trans} to obtain quadratic transmuted distributions and verified on real data that these latter fit the data better than the baseline distributions. Later, \citet{granzotto-et-al-2017} showed that the formula \eqref{eq:quad-trans} represents the cdf of a mixture of order statistics. But this quadratic transmutation does not adapt to certain complex data \citep{rahman-et-al-2020a,saha-et-al-2024}. Thus, in order to obtain more flexibility, \citet{granzotto-et-al-2017} then extended it to develop cubic transmutation (CT) through the new cdf
\begin{equation}
	F_G(x) = \lambda_1 G(x) + (\lambda_2 - \lambda_1) G^2(x) + (1 - \lambda_2) G^3(x),
\end{equation}
where $\lambda_1$ and $\lambda_2$ are additional parameters such that $0 \leqslant \lambda_1 \leqslant 1$ and $-1 \leqslant \lambda_2 \leqslant 1$. They applied the CT to Weibull and log-logistic distributions on real data and found that the cubic transmutations presented better fits to the data than the quadratic transmutation and the baseline distribution.

Following \citet{granzotto-et-al-2017}, several authors have developed other cubic transmutation formulas. We can mention \citet{alkadim-mohammed-2017}, \citet{rahman-et-al-2018a}, \citet{rahman-et-al-2018b}, \cite{rahman-et-al-2019b} and \citet{rahman-et-al-2023}. For each formula, the authors verified using real data that the proposed cubic transmutation fit the data better than the quadratic transmutation. Some of the cubic transmutation formulas have been compared to each other but all these six cubic transmutation formulas have yet to be all compared with each other. Some have been compared individually to the formulas of \cite{granzotto-et-al-2017} and/or \cite{alkadim-mohammed-2017} for Weibull, Pareto, log-logistic, power function, Kumaraswamy distributions \citep{aslam-et-al-2020,rahman-et-al-2019a,rahman-et-al-2021,rahman-et-al-2024,hafeez-et-al-2023,tushar-et-al-2024}. \citet{sakthivel-vidhya-2023} applied four of these six cubic transmutation formulas on Rayleigh distribution, gave their statistical properties and compared the four different formulas for fitting three datasets. However, the four formulas were not simultaneously compared with each other on the same data set. Moreover, the theoretical aspects of the comparison were also not considered.

The objective of this paper is to provide a comparative analysis of different cubic transmutations and to present a case study using cubic transmutations of the Pareto distribution. The paper is organized as follows. In Section \ref{sec:review-CT}, we review the state of the art regarding cubic transmutation formulas. Section \ref{sec:results} presents a critical analysis of the various formulas, focusing particularly on the parameter ranges. This analysis will lead to the proposal of modified parameter ranges that extend the obtained distributions. We will also examine the formulas to determine if some of the studied families of distributions can be obtained as subfamilies of the others. Section \ref{sec:illustration} presents an illustration of our results by comparing the performances of different cubic transmuted Pareto distributions in fitting real data in R software \citep{R-2024}. The paper concludes with a discussion and final remarks in Section \ref{sec:conclusion}.

\section{Cubic transmutation: a review of existing formulas}
\label{sec:review-CT}

Let $x \in \reels$ and $G(x)$ be a cumulative density function (cdf) of a given continuous random variable. It is possible that $G(x)$ depends on one real parameter $\xi \in \reels$ or a parameter vector $\xi \in \reels^d$, where $d \geqslant 2$. But for the sake of notation, we will ignore $\xi$ in the notation. According to the literature \citep[see for example][]{rahman-et-al-2020b,ali-athar-2021,imliyangba-et-al-2021}, two methods have been used to obtain a new cdf $F(x)$ called "cubic transmutation" of $G(x)$: 
\begin{list}{$\bullet$}{}
	\item the first one is based on the use of order statistics \citep{granzotto-et-al-2017,rahman-et-al-2018a,rahman-et-al-2018b};
	\item the second one \citep{alkadim-mohammed-2017,rahman-et-al-2019b,rahman-et-al-2023} is based on the use of the following formula which is a special case of a more general formula proposed by \citet{alzaatreh-et-al-2013}:
	\begin{equation}
		\label{eq:cdf-alzaatreh}
		F(x) = R(G(x)) = \int_{0}^{G(x)} r(t) \dt,
	\end{equation}
	where $R(t)$ is a polynomial cdf of degree 3 increasing, right-continuous such that $R(0)=0$ and $R(1)=1$ and $r(t)$ is the polynomial probability density function (pdf) of degree 2 linked to $R(t)$.
\end{list}

Note that if $g(x)$ is the pdf linked to $G(x)$, then the pdf linked to the cdf $F(x)$ is 
\begin{equation}
	\label{eq:pdf-alzaatreh}
	f(x) = F'(x) = g(x) \, r(G(x)).
\end{equation}

In the remainder of this section, we review six cubic transmutation formulas based on one or the other of the two methods.

\subsection{First formula: \citet{granzotto-et-al-2017}}

According to \citet{saha-et-al-2024}, the first cubic transmutation formula is due to \citet{granzotto-et-al-2017}. The latter authors proved that the quadratic transmutation formula \eqref{eq:quad-trans} of \citet{shaw-buckley-2009} could be rewritten as the cdf of a mixture of order statistics. Indeed, if $X_1$ and $X_2$ are two independent random variables sampled from the probability distribution with cdf $G(x)$ and if $X_{1:2}$, $X_{2:2}$ represent the corresponding order statistics, then the mixture random variable
\begin{equation} 
	X = \begin{cases} 
		X_{1:2} & \text{with probability $\pi$}, \\
		X_{2:2} & \text{with probability $1-\pi$},
	\end{cases}
\end{equation}
where $0 \leqslant \pi \leqslant 1$, has the cdf 
\begin{align*} 
	F_S(x) & = \pi \proba(X_{1:2} \leqslant x) + (1-\pi) \proba(X_{2:2} \leqslant x) \\
	& = \pi \left[ 1 - (1-G(x))^2 \right] + (1-\pi) G^2(x) \\
	& = 2\pi G(x) + (1-2\pi) G^2(x) \\
	& = (1+\lambda) G(x) - \lambda G^2(x),
\end{align*} 
where $\lambda = 2\pi-1$ is such that $-1 \leqslant \lambda \leqslant 1$. They extended this reasoning to three variables to obtain their cubic transmutation formula. Indeed, if $X_1$, $X_2$ and $X_3$ are three independent random variables sampled from the probability distribution with cdf $G(x)$ and if $X_{1:3}$, $X_{2:3}$, $X_{3:3}$ represent the corresponding order statistics, then the mixture random variable
\begin{equation} 
	X = \begin{cases} 
		X_{1:3} & \text{with probability $\pi_1$}, \\
		X_{2:3} & \text{with probability $\pi_2$}, \\
		X_{3:3} & \text{with probability $\pi_3$},
	\end{cases}
\end{equation}
where for all $i=1,2,3$, $0 \leqslant \pi_i \leqslant 1$ and $\pi_1 + \pi_2 + \pi_3 = 1$, has the cdf 
\begin{equation*} 
	F_G(x) = \pi_1 \proba(X_{1:3} \leqslant x) + \pi_2 \proba(X_{2:3} \leqslant x) + \pi_3 \proba(X_{3:3} \leqslant x),
\end{equation*}
where, according to the formula of the cdf of order statistics (see for example, \citet{david-mishriky-1968}), 
\begin{equation} 
	\proba(X_{i:3} \leqslant x) = \sum_{r=i}^{3} C^r_3 G^r(x) (1-G(x))^{3-r}.
\end{equation}
So, \citet{granzotto-et-al-2017} obtained
\begin{equation}
	\label{eq:cdf-granz-0}
	F_G(x) = 3\pi_1 G(x) + 3(\pi_2 - \pi_1) G^2(x) + (1 - 3\pi_2) G^3(x)
\end{equation}
which may be rewritten
\begin{equation}
	\label{eq:cdf-granz} 
	F_G(x) = \lambda_1 G(x) + (\lambda_2 - \lambda_1) G^2(x) + (1 - \lambda_2) G^3(x),
\end{equation} 
where $\lambda_1 = 3\pi_1$ and $\lambda_2 = 3\pi_2$. Notice that setting $\lambda_1 = \lambda_2 = 1$ enables to get the baseline distribution with cdf $G(x)$. The range proposed by \citet{granzotto-et-al-2017} for parameters $\lambda_1$ and $\lambda_2$ is 
\begin{equation}
	\label{eq:range-granz} 
	\mathscr{S}_G = \Big\{ (\lambda_1,\lambda_2) : 0 \leqslant \lambda_1 \leqslant 1 \;\; \text{and} \;\; -1 \leqslant \lambda_2 \leqslant 1 \Big\}.
\end{equation}
The corresponding pdf is 
\begin{equation}
	\label{eq:pdf-granz} 
	f_G(x) = g(x) \left[ \lambda_1 + 2(\lambda_2 - \lambda_1) G(x) + 3(1 - \lambda_2) G^2(x) \right].
\end{equation}
The cubic transmutation of \citet{granzotto-et-al-2017} will be denoted $CT_G$.

\begin{remark}
	It is worth mentioning that the $CT_G$ is not a generalization of quadratic transmutation. Indeed, the formula \eqref{eq:quad-trans} can be obtained from the formula \eqref{eq:cdf-granz} only by setting $\lambda_1 = 1+\lambda$ and $\lambda_2 = 1$. By doing so, we would have $0 \leqslant 1+\lambda \leqslant 1$ or, equivalently, $-1 \leqslant \lambda \leqslant 0$ which is different from the parameter range of quadratic transmutation. 
\end{remark}

\begin{remark}
	\citet{aslam-et-al-2018} extended the cdf formula \eqref{eq:cdf-granz} by replacing $G(x)$ by $G(x)^{\alpha}$, where $\alpha > 0$. They obtained a new family which they also called transmuted family of distributions. However this can be considered as a compounding of transmutation and exponentiation \citep{ali-et-al-2007} so it will not be considered in this study. \citet{riffi-2019} also proposed new families of distributions called higher rank transmutation (HRT) maps to generate more flexible and tractable lifetime families of probability models. \citet{merovci-et-al-2016} introduced a new family called another generalized transmuted (AGT) family of distributions by replacing $G(x)$ by $1-(1-G(x))^{\alpha}$ in \eqref{eq:quad-trans}. \citet{nofal-et-al-2017} proposed the generalized transmuted (GT) family of distributions by replacing $G(x)$ by $G^a(x)$ and $G^2(x)$ by $G^{a+b}(x)$ in \eqref{eq:quad-trans}, where $a$ and $b$ are positive parameters. But as for \citet{aslam-et-al-2018}, AGT, HRT and GT can be considered as compounding of transmutation and exponentiation so they do not fit with the scope of this study. 
\end{remark}

\subsection{Second formula: \citet{alkadim-mohammed-2017}}

The cubic transmutation of \citet{alkadim-mohammed-2017} (that we will denote $CT_{A}$) corresponds to the definition of the cdf 
$$
R(t) = (1+\lambda) t - 2\lambda t^2 + \lambda t^3
$$
for $t \in [0,1]$ and the use of Formula \eqref{eq:cdf-alzaatreh}. For a baseline cdf $G(x)$, it corresponds to the cdf  
\begin{equation}
	\label{eq:cdf-alkdm}
	F_A(x) = (1+\lambda) G(x) - 2\lambda G^2(x) + \lambda G^3(x),
\end{equation}
where $-1 \leqslant \lambda \leqslant 1$. The corresponding pdf is
\begin{equation}
	\label{eq:pdf-alkdm}
	f_A(x) = g(x) \left[ (1+\lambda) - 4\lambda G(x) + 3\lambda G^2(x) \right].
\end{equation}
Setting $\lambda=0$ enables to get the baseline distribution. It is important to note that this form of cubic transmutation is not a generalization of quadratic transmutation because the term $\lambda G^3(x)$ can only vanish if $\lambda=0$, which would also force the term $2\lambda G^2(x)$ to vanish.

\subsection{Third formula: \citet{rahman-et-al-2018a}}

By permuting the probabilities $\pi_1$ and $\pi_3$ in Equation \eqref{eq:cdf-granz-0} and taking into account the relation $\pi_1 + \pi_2 + \pi_3 = 1$, \citet {rahman-et-al-2018a} obtained another form of cubic transmutation whose cdf is defined by
\begin{equation}
	\label{eq:cdf-granz-0-b}
	F(x) = (3 - 3\pi_1 - 3\pi_2) G(x) + (3\pi_1 + 6\pi_2 - 3) G^2(x) - (3\pi_2 - 1) G^3(x).
\end{equation}
By setting $\lambda_1 = 2 - 3\pi_1 - 3\pi_2$ and $\lambda_2 = 3\pi_2 - 1$, they obtained
\begin{equation}
	\label{eq:cdf-r18a}
	F_{R18a}(x) = (1 + \lambda_1) G(x) + (\lambda_2 - \lambda_1) G^2(x) - \lambda_2 G^3(x),
\end{equation}
where the couple of new parameters $(\lambda_1,\lambda_2)$ belongs to the set
\begin{equation}
	\label{eq:range-r18a}
	\mathscr{S}_{R18a} = \Big\{ (\lambda_1,\lambda_2) :-1 \leqslant \lambda_1 \leqslant 1, \;\; -1 \leqslant \lambda_2 \leqslant 1 \;\; \text{and} \;\; -2 \leqslant \lambda_1 + \lambda_2 \leqslant 1 \Big\}.
\end{equation}
The corresponding pdf is 
\begin{equation}
	\label{eq:pdf-r18a}
	f_{R18a}(x) = g(x) \left[ (1 + \lambda_1) + 2(\lambda_2 - \lambda_1) G(x) - 3\lambda_2 G^2(x) \right].
\end{equation}
One gets the baseline distribution by setting $\lambda_1 = \lambda_2 = 0$ and the quadratic transmutation by setting $\lambda_1 = \lambda$ and $\lambda_2 = 0$.

\subsection{Fourth formula: \citet{rahman-et-al-2018b}}

Considering Equation \eqref{eq:cdf-granz-0-b} and setting $\lambda_1 = 1 - 3\pi_1$ and $\lambda_2= 1 - 3\pi_2$, \citet{rahman-et-al-2018b} obtained a new distribution function in the form 
\begin{equation}
	\label{eq:cdf-r18b}
	F_{R18b}(x) = (1 + \lambda_1 + \lambda_2) G(x) - (\lambda_1 + 2\lambda_2) G^2(x) + \lambda_2 G^3(x),
\end{equation}
and the corresponding pdf is 
\begin{equation}
	\label{eq:pdf-r18b}
	f_{R18b}(x) = g(x) \left[ (1 + \lambda_1 + \lambda_2) - 2(\lambda_1 + 2\lambda_2) G(x) + 3\lambda_2 G^2(x) \right].
\end{equation}
One gets the baseline distribution by setting $\lambda_1 = \lambda_2 = 0$ and the quadratic transmutation by setting $\lambda_1 = \lambda$ and $\lambda_2 = 0$.

The range of the parameters are
\begin{equation}
	\label{eq:range-r18b}
	\mathscr{S}_{R18b} = \Big\{ (\lambda_1,\lambda_2) : -1 \leqslant \lambda_1 \leqslant 1 \; \text{and} \; 0 \leqslant \lambda_2 \leqslant 1 \Big\}.
\end{equation}

\subsection{Fifth formula: \citet{rahman-et-al-2019b}}

Using the formula \eqref{eq:cdf-alzaatreh} for the pdf
$$
r(t) = (1-\lambda) + 6\lambda t - 6 \lambda t^2
$$
on $[0,1]$, \citet{rahman-et-al-2019b} defined another form of cubic transmutation with the following cdf:
\begin{equation}
	\label{eq:cdf-r19}
	F_{R19}(x) = (1 - \lambda) G(x) + 3\lambda G^2(x) - 2\lambda G^3(x),
\end{equation}
where $-1 \leqslant \lambda \leqslant 1$. The corresponding pdf is
\begin{equation}
	\label{eq:pdf-r19}
	f_{R19}(x) = g(x) \left[ (1 - \lambda) + 6\lambda G(x) - 6\lambda G^2(x) \right].
\end{equation}
This cubic transmutation formula provides the baseline distribution for $\lambda = 0$. For the same reasons as those mentioned for the formula of \citet{alkadim-mohammed-2017}, this form of cubic transmutation is not a generalization of quadratic transmutation.

\begin{remark}
	The transmutation formula \eqref{eq:cdf-r19} has also been considered by \citet{aslam-et-al-2020} and called modified transform (MT) distribution. The latter studied the statistical characteristics of MT distributions and applied it using Weibull distribution as baseline distribution. Application on two real datasets suggest that their proposed MT Weibull distribution fits the data better than the CT Weibull using the formula of \citet{granzotto-et-al-2017}. 
\end{remark}

\subsection{Sixth formula: \citet{rahman-et-al-2023}}

Using the formula \eqref{eq:cdf-alzaatreh} for the pdf 
$$
r(t) = 1 - \lambda (\eta - 1) + 2\lambda (2\eta - 1)t - 3\lambda \eta t^2
$$
on $[0,1]$, \citet{rahman-et-al-2023} have developed a modified cubic transmuted family of distributions. The cdf of this family of distributions is defined for all $x \in \reels$ by
\begin{equation}
	\label{eq:cdf-r23} 
	F_{R23}(x) = \left[ 1 - \lambda (\eta - 1) \right] G(x) + \lambda (2\eta - 1) G^2(x) - \lambda \eta G^3(x),
\end{equation}
where 
\begin{equation} 
	\label{eq:range-r23}
	(\lambda,\eta) \in \mathscr{S}_{R23} = \Big\{ (\lambda, \eta) : -1 \leqslant \lambda \leqslant 1 \;\; \text{and} \;\; 0 \leqslant \eta \leqslant 2 \Big\}
\end{equation}
and the corresponding pdf is
\begin{equation} 
	\label{eq:pdf-r23}
	f_{R23}(x) = g(x) \left[ 1 - \lambda (\eta - 1) + 2\lambda (2\eta - 1) G(x) - 3 \lambda \eta G^2(x) \right].
\end{equation}
When $\lambda = 0$, one gets the baseline distribution. Setting $\eta = 0$ enables to get the quadratic transmutation of \citet{shaw-buckley-2009} and setting $\eta = 2$ enables to get the cubic transmutation of \citet{rahman-et-al-2019b}.

\section{Main theoretical comparison results}
\label{sec:results}

For the six cubic transmutation models presented in Section \ref{sec:review-CT}, the constraints on the new parameters are very important to ensure that the new cdf verify all the properties required for a cumulative distribution function. For each model, the parameter range (or the set of all possible values of the parameters) are as important as the definition set would be for a function. In this sense, the comparison of two transmutation formulas must be done not only in terms of expression of cdf but also in terms of range of parameters. In this section, we give some theoretical comparison results on these different formulas. In the straight line of Definition \ref{def:CT}, we carry out a critical study of different formulas, in particular by re-examining the parameter ranges in such a way as to always guarantee that the new distribution is well defined (i.e. is truly a probability distribution). As suggested by some comments and calculations in \citep{hameldarbandi-yilmaz-2020,yilmaz-2025}, the theoretical study of the range of parameters is far from easy when the cubic transmutation contains two additional parameters and it is possible to then resort to a graphical study. The study made in this section will lead us to propose modified versions of the parameter range allowing us to extend the distributions obtained. We will also study the formulas to determine if some of the studied families of distributions can be obtained as subfamilies of the others. 

Our main theoretical comparison results are given by Proposition \ref{prop:granz} to Proposition \ref{prop:ra19} hereafter.

\begin{proposition}
	\label{prop:granz}
	Let $G(x)$ be a cdf of a given random variable. The Granzotto Cubic Transmuted distribution of parameters $\lambda_1$ and $\lambda_2$ (which we will denote $CT_G(\lambda_1,\lambda_2)$) with cdf $F_G(x)$ (see Equation \eqref{eq:cdf-granz}) is not well defined under conditions \eqref{eq:range-granz}. The $CT_G(\lambda_1,\lambda_2)$ distribution is rather well defined under the following modified conditions on $\lambda_1$ and $\lambda_2$:
	\begin{equation} 
		\label{eq:range-granz-b} 
		\mathscr{S}_{MG} = \Big\{ (\lambda_1,\lambda_2) : 0 \leqslant \lambda_1 \leqslant 3, \; 0 \leqslant \lambda_2 \leqslant 3 \; \text{and} \; 0 \leqslant \lambda_1 + \lambda_2 \leqslant 3 \Big\},
	\end{equation}
	and will be called modified Granzotto Cubic Transmutation and denoted $CT_{MG}$.
\end{proposition}

\begin{proof}
	By definition of the cdf of order statistics (see for example, \citet{david-mishriky-1968}), the cdf 
	$$
	F_G(x) = 3\pi_1 G(x) + 3(\pi_2 - \pi_1) G^2(x) + (1 - 3\pi_2) G^3(x)
	$$
	from Equation \eqref{eq:cdf-granz-0} is well defined under the conditions $0 \leqslant \pi_i \leqslant 1$ for all $i=1,2,3$ and $\pi_1 + \pi_2 = 1-\pi_3 \in [0,1]$. Taking into account the relationships $\lambda_1 = 3\pi_1$ and $\lambda_2 = 3\pi_2$, the parameters $\lambda_1$ and $\lambda_2$ therefore satisfy the three constraints \eqref{eq:range-granz-b}. Since the conditions \eqref{eq:range-granz} defined by \citet{granzotto-et-al-2017} are not included in the conditions \eqref{eq:range-granz-b}, it is enough to find a counter-example to show that $CT_G$ is not well defined under conditions \eqref{eq:range-granz}. Let us set $\lambda_1 = 0$ and $\lambda_2 = -0.5$. The couple $(\lambda_1, \lambda_2) = (0,-0.5)$ satisfies conditions \eqref{eq:range-granz} and the corresponding cdf is 
	\begin{equation*}
		F_G(x) = -0.5 \, G^2(x) + 1.5 \, G^3(x) = 1.5 \, G^2(x) \left[ G(x) - \frac{1}{3} \right].
	\end{equation*}
	Let $q_{1/3}$ be the quantile of order $1/3$ of the baseline distribution associated with the cdf $G(x)$. We have
	$$
	\begin{cases}
		1.5 \, G^2(x) \left[ G(x) - \frac{1}{3} \right] \leqslant 0 & \text{if} \quad x \leqslant q_{1/3}, \\[0.2cm]
		1.5 \, G^2(x) \left[ G(x) - \frac{1}{3} \right] \geqslant 0 & \text{if} \quad x \geqslant q_{1/3},
	\end{cases}
	$$
	which proves that $F_G(x)$ is not a cumulative distribution function for $\lambda_1 = 0$ et $\lambda_2 = -0.5$.
\end{proof}

\begin{remark}
	By relaxing the assumption of mixture density of order statistics which requires that $\lambda_1 \geqslant 0$ and $\lambda_2 \geqslant 0$ and by using the method based on Formula \eqref{eq:cdf-alzaatreh}, the fact that $f_G(x)$ is a well defined pdf is equivalent to the fact that the function 
	\begin{equation} 
		\label{eq:r-granz}
		r(t) = \lambda_1 + 2(\lambda_2 - \lambda_1) t + 3(1 - \lambda_2) t^2
	\end{equation}
	is positive and thus a pdf on $[0,1]$. Figure \ref{fig:verif-granz} presents a graphical study (for $\lambda_1 \in [-0.5,4.5]$ and $\lambda_2 \in [-3.5,3.5]$ with discretization step equal to 0.1 for both parameters) of whether or not, $r(t)$ defined by Equation \eqref{eq:r-granz} is a positive function on $[0,1]$. In this figure, the range \eqref{eq:range-granz} proposed by \citet{granzotto-et-al-2017} is represented by a square area delimited in red dotted lines while the modified range \eqref{eq:range-granz-b} proposed in this paper is represented by the black triangular area. We notice that there exist values of $\lambda_1$ and $\lambda_2$ outside the conditions \eqref{eq:range-granz-b} for which $r(t)$ is positive but the shape of the graph does not suggest any simple expression of the constraints on these parameters apart from the conditions \eqref{eq:range-granz-b}.
	
	\begin{figure}[!h]
		\centering
		\includegraphics[scale=0.5]{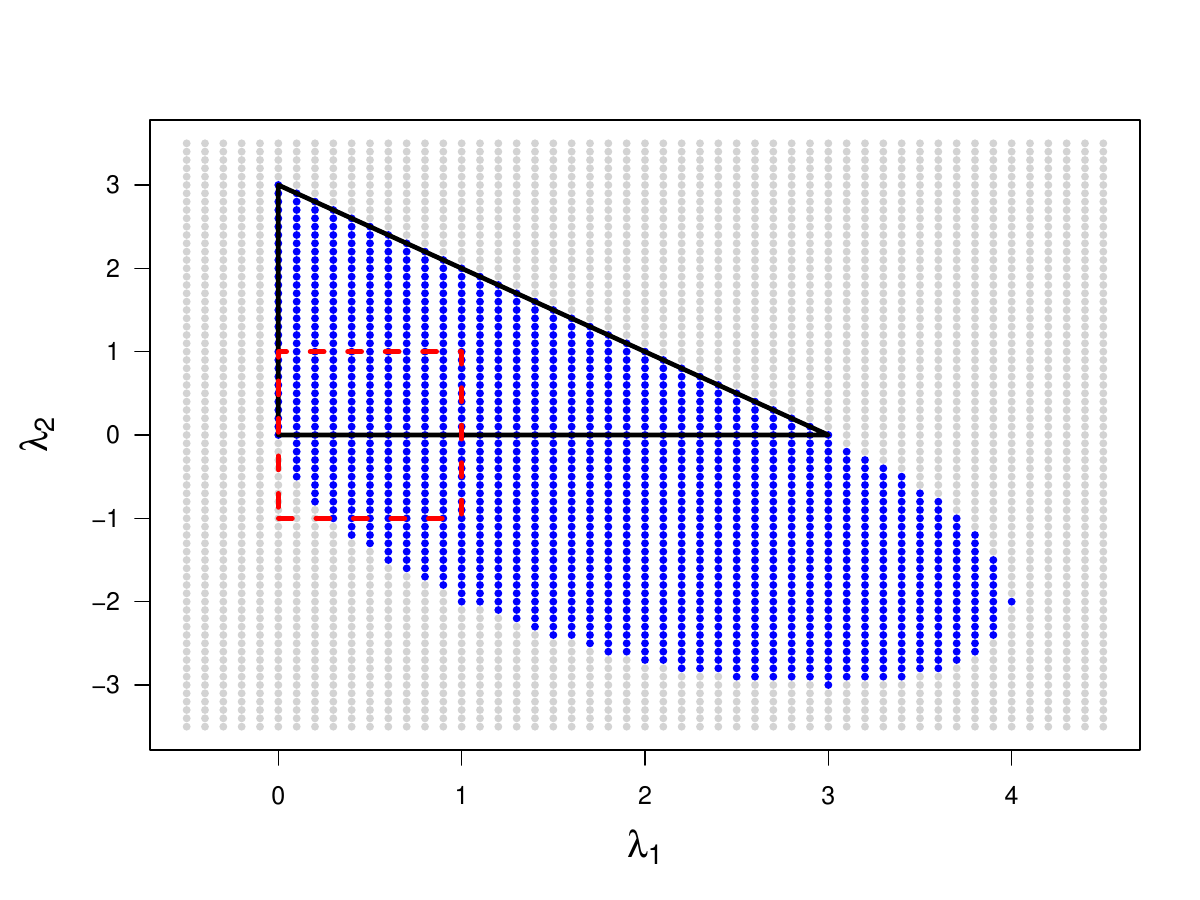}
		\caption{Graphical study for different values of $\lambda_1$ and $\lambda_2$ of whether or not $r(t) = \lambda_1 + 2(\lambda_2 - \lambda_1) t + 3(1 - \lambda_2) t^2$ is positive on $[0,1]$ (blue colour means "yes" and gray means "no").}
		\label{fig:verif-granz}
	\end{figure}
\end{remark}

\begin{remark}
	The fact that the Granzotto Cubic Transmuted distribution is not well defined has also been pointed out by \citet{hameldarbandi-yilmaz-2020}. However, the proof and the modified range given in our paper are different from the ones given by these latter authors. Indeed, the approach used by \citet{hameldarbandi-yilmaz-2020} is based on a reparametrization leading to the cubic transmutation of \citet{rahman-et-al-2018a} while the approach used in the present paper is more direct and more simple.
\end{remark}

\begin{remark}
	Since the Granzotto cubic transmuted distribution is not well defined, it does not make sense to compare other CT distributions with it but rather with its modified version.
\end{remark}

\begin{proposition}
	\label{prop:alk}
	Let $G(x)$ be the cdf of a random variable. 
	\begin{enumerate}[i)]
		\item For all $\lambda$ such that $-1 \leqslant \lambda \leqslant 1$, the $CT_A(\lambda)$ family linked to $G(x)$ is a sub-family of the $CT_{R18a}(\lambda_1,\lambda_2)$ family of distributions linked to $G(x)$.
		\item The $CT_A(\lambda)$ family of distributions defined by Equation \eqref{eq:cdf-alkdm} is not a sub-family of the $CT_G(\lambda_1,\lambda_2)$ family of distributions if $\lambda \neq 0$.
		\item The $CT_A(\lambda)$ family of distributions is rather a sub-family of the $CT_{MG}(\lambda_1,\lambda_2)$ family of distributions for $-1 \leqslant \lambda \leqslant 1$. 
		\item Moreover, the $CT_A(\lambda)$ distribution linked to $G(x)$ is well defined under the extended condition $-1 \leqslant \lambda \leqslant 3$ and will be called modified AL-Kadim Cubic Transmutation and denoted $CT_{MA}(\lambda)$.
	\end{enumerate}
\end{proposition}

\begin{proof} 
	\mbox{}
	\begin{enumerate}[i)]
		\item Let $\lambda \in [-1,1]$. Equation \eqref{eq:pdf-alkdm} can be rewritten under the form of Equation \eqref{eq:pdf-r18a} by setting $\lambda_1 = \lambda$ and $\lambda_2 = -\lambda$. We then have $\lambda_1 = \lambda \in [-1,1]$, $\lambda_2 = -\lambda \in [-1,1]$ and $\lambda_1 + \lambda_2 = \lambda - \lambda = 0 \in [-2,1]$ so we can conclude that the $CT_A(\lambda)$ family is a sub-family of the $CT_{R18a}(\lambda_1,\lambda_2)$ family of distributions.
		
		\item Let us first prove by absurd that the $CT_A(\lambda)$ family of distributions is not a sub-family of the $CT_G(\lambda_1,\lambda_2)$ family if $\lambda \neq 0$. Certainly, as noted by \citet{saha-et-al-2024}, Equation \eqref{eq:pdf-alkdm} can be rewritten under the form of Equation \eqref{eq:pdf-granz} by setting $\lambda_1 = 1+\lambda$ and $\lambda_2 = 1-\lambda$. However, there is a small subtlety that makes all the difference. It is that under such a change of variables, we would have 
		\begin{align*} 
			\left. \begin{array}{c} 
				0 \leqslant \lambda_1 \leqslant 1 \\ 
				-1 \leqslant \lambda_2 \leqslant 1 
			\end{array} \right\} & \iff \left\{ \begin{array}{c} 
				0 \leqslant 1+\lambda \leqslant 1 \\ 
				-1 \leqslant 1-\lambda \leqslant 1 
			\end{array} \right. \\
			& \iff \left\{ \begin{array}{c} 
				-1 \leqslant \lambda \leqslant 0 \\
				0 \leqslant \lambda \leqslant 2 
			\end{array} \right. \\
			& \iff \lambda=0.
		\end{align*}
		Thus, the $CT_A(\lambda)$ family of distributions is a sub-family of the $CT_G(\lambda_1,\lambda_2)$ family only if $\lambda = 0$.
		
		\item Let us now prove that the $CT_A(\lambda)$ family of distributions is a sub-family of the $CT_{MG}$ family of distributions for all $\lambda \in [-1,1]$. Let $\lambda_1 = 1+\lambda$ and $\lambda_2 = 1-\lambda$. 
		\begin{align*} 
			\left. \begin{array}{c} 
				0 \leqslant \lambda_1 \leqslant 3 \\ 
				0 \leqslant \lambda_2 \leqslant 3 
			\end{array} \right\} & \iff \left\{ \begin{array}{c} 
				0 \leqslant 1+\lambda \leqslant 3 \\ 
				0 \leqslant 1-\lambda \leqslant 3 
			\end{array} \right. \\
			& \iff \left\{ \begin{array}{c} 
				-1 \leqslant \lambda \leqslant 2 \\
				-2 \leqslant \lambda \leqslant 1 
			\end{array} \right. \\
			& \iff -1 \leqslant \lambda \leqslant 1.
		\end{align*}
		
		\item Let us now prove that the $CT_A(\lambda)$ distribution is well defined under the extended condition $-1 \leqslant \lambda \leqslant 3$. The pdf \eqref{eq:pdf-alkdm} is given by 
		$$
		f_A(x) = g(x) \, r(G(x)),
		$$
		where for all $t \in [0,1]$, $r(t) = (1+\lambda) - 4\lambda t + 3\lambda t^2$ and, according to \cite{alzaatreh-et-al-2013}, $f_A(x)$ is well defined if $r(t)$ is a pdf on $[0,1]$. We have
		$$
		\int_{0}^{1} r(t) \dt = 1.
		$$
		So for $r(t)$ to be a pdf, it is enough to check that for all $t \in [0,1]$, $r(t) \geqslant 0$. We have $r'(t) = 2\lambda (3t-2)$. 
		\begin{list}{$\circ$}{}
			\item If $\lambda \leqslant 0$, then the variation table of $r(t)$ is as follows: 
			$$
			\begin{tabvar}{|C|LCCCC|} 
				\hline t & 0 & & \frac{2}{3} & & 1 \\ 
				\hline r'(t) & & + & \barre{0} & - & \\ 
				\hline \niveau{1}{2} \TVcenter{r(t)} & 1+\lambda & \croit & \TVstretch[3pt]{\frac{3-\lambda}{3}} & \decroit & 1 \\ 
				\hline		
			\end{tabvar}
			$$
			It is enough that $\lambda \geqslant -1$.
			
			\item If $\lambda \geqslant 0$, then the variation table of $r(t)$ is as follows:
			$$
			\begin{tabvar}{|C|LCCCC|} 
				\hline t & 0 & & \frac{2}{3} & & 1 \\ 
				\hline r'(t) & & - & \barre{0} & + & \\ 
				\hline \niveau{2}{2} \TVcenter{r(t)} & \TVstretch[3pt]{1+\lambda} & \decroit & \TVstretch[3pt]{\frac{3-\lambda}{3}} & \croit & 1 \\ 
				\hline		
			\end{tabvar}
			$$
			It is enough that $\lambda \leqslant 3$.
		\end{list}
		Thus, $r(t)$ is positive on $[0,1]$ if $-1 \leqslant \lambda \leqslant 0$ or $0 \leqslant \lambda \leqslant 3$. We conclude that $r(t)$ is positive on $[0,1]$ if $-1 \leqslant \lambda \leqslant 3$.
	\end{enumerate}	
\end{proof}

\begin{proposition}
	\label{prop:ra18a}
	Let $G(x)$ be the cdf of a random variable. 
	\begin{enumerate}[i)]
		\item The cubic transmuted family of distributions of parameters $\lambda_1$ and $\lambda_2$ proposed by \citet{rahman-et-al-2018a} (denoted $CT_{R18a}(\lambda_1,\lambda_2)$) linked to $G(x)$ is well defined under conditions \eqref{eq:range-r18a}. 
		\item The $CT_{R18a}(\lambda_1,\lambda_2)$ family is also well defined under the extended range
		\begin{equation} 
			\label{eq:range-r18a-b}
			\mathscr{S}_{MR18a} = \Big\{ (\lambda_1,\lambda_2) : -1 \leqslant \lambda_1 \leqslant 2, \; -1 \leqslant \lambda_2 \leqslant 2 \; \text{and} \; -2 \leqslant \lambda_1 + \lambda_2 \leqslant 1 \Big\}
		\end{equation}
		and will be called modified $CT_{R18a}$ and denoted $CT_{MR18a}$.
		\item For all $(\lambda_1,\lambda_2) \in \mathscr{S}_{MR18a}$, the $CT_{MR18a}(\lambda_1,\lambda_2)$ family of distributions linked to $G(x)$ is equivalent to the $CT_{MG}(\lambda_1 + 1, \lambda_2 + 1)$ family of distributions linked to $G(x)$.
		\item For all $(\lambda_1,\lambda_2) \in \mathscr{S}_{MG}$, the $CT_{MG}(\lambda_1,\lambda_2)$ family of distributions linked to $G(x)$ is equivalent to the $CT_{MR18a}(\lambda_1 - 1, \lambda_2 - 1)$ family of distributions linked to $G(x)$.
	\end{enumerate}
\end{proposition}

\begin{proof}
	We start the proof with item ii).
	\begin{enumerate}[i)]
		\item[ii)] By definition of the cdf of order statistics, the cdf 
		$$
		F(x) = (3 - 3\pi_1 - 3\pi_2) G(x) + (3\pi_1 + 6\pi_2 - 3) G^2(x) - (3\pi_2 - 1) G^3(x)
		$$
		from Equation \eqref{eq:cdf-granz-0} is well defined under the conditions $0 \leqslant \pi_i \leqslant 1$ for all $i=1,2,3$ and $\pi_1 + \pi_2 = 1-\pi_3 \in [0,1]$. Since $\lambda_1 = 2 - 3\pi_1 - 3\pi_2$ and $\lambda_2 = 3\pi_2 - 1$, we can write $\lambda_1 = 2 - 3(1-\pi_3) = 3\pi_3 - 1$, $\lambda_2 = 3\pi_2 - 1$ and $\lambda_1 + \lambda_2 = 1-3\pi_1$. Taking into account the range of $\pi_1$, $\pi_2$ and $\pi_3$, we deduce the three constraints \eqref{eq:range-r18a-b}. 
		
		\item[i)] It is immediate from ii) since the conditions \eqref{eq:range-r18a} set by \citet{rahman-et-al-2018a} on $\lambda_1$ and $\lambda_2$ are included in conditions \eqref{eq:range-r18a-b}.
		
		\item[iii)] Let $(\lambda_1,\lambda_2) \in \mathscr{S}_{MR18a}$. It is easily seen that by replacing $\lambda_1$ by $\lambda_1' = 1+\lambda_1$ and $\lambda_2$ by $\lambda_2' = 1+\lambda_2$ in \eqref{eq:cdf-granz}, we obtain \eqref{eq:cdf-r18a}. We also easily verify that $(\lambda_1,\lambda_2) \in \mathscr{S}_{MR18a}$ if and only if $(1+\lambda_1,1+\lambda_2) \in \mathscr{S}_{MG}$. Indeed, 
		\begin{align*}
			(\lambda_1,\lambda_2) \in \mathscr{S}_{MR18a} & \iff \left\{ \begin{array}{c}
				-1 \leqslant \lambda_1 \leqslant 2 \\ -1 \leqslant \lambda_2 \leqslant 2 \\ -2 \leqslant \lambda_1 + \lambda_2 \leqslant 1 
			\end{array} \right. \\
			& \iff \left\{ \begin{array}{c}
				0 \leqslant \lambda_1 + 1 \leqslant 3 \\ 0 \leqslant \lambda_2 + 1 \leqslant 3 \\ 0 \leqslant \lambda_1 + \lambda_2 + 2 \leqslant 3
			\end{array} \right. \\
			& \iff (\lambda_1 + 1, \lambda_2 + 1) \in \mathscr{S}_{MG}
		\end{align*}
		which completes the proof because $(\lambda_1 + 1) + (\lambda_2 + 1) = \lambda_1 + \lambda_2 + 2$.
		
		\item[iv)] Let $(\lambda_1,\lambda_2) \in \mathscr{S}_{MG}$. It is easily seen that by replacing $\lambda_1$ by $\lambda_1' = \lambda_1 - 1$ and $\lambda_2$ by $\lambda_2' = \lambda_2 - 1$ in \eqref{eq:cdf-r18a}, we obtain \eqref{eq:cdf-granz}. We also easily verify that $(\lambda_1,\lambda_2) \in \mathscr{S}_{MG}$ if and only if $(\lambda_1 - 1, \lambda_2 - 1) \in \mathscr{S}_{MR18a}$. Indeed, 
		\begin{align*}
			(\lambda_1,\lambda_2) \in \mathscr{S}_{MG} & \iff \left\{ \begin{array}{c}
				0 \leqslant \lambda_1 \leqslant 3 \\ 0 \leqslant \lambda_2 \leqslant 3 \\ 0 \leqslant \lambda_1 + \lambda_2 \leqslant 3 
			\end{array} \right. \\
			& \iff \left\{ \begin{array}{c}
				-1 \leqslant \lambda_1 - 1 \leqslant 2 \\ -1 \leqslant \lambda_2 - 1 \leqslant 2 \\ -2 \leqslant \lambda_1 + \lambda_2 - 2 \leqslant 1
			\end{array} \right. \\
			& \iff (\lambda_1 - 1, \lambda_2 - 1) \in \mathscr{S}_{MR18a}
		\end{align*}
		which completes the proof because $(\lambda_1 - 1) + (\lambda_2 - 1) = \lambda_1 + \lambda_2 - 2$.
	\end{enumerate}
\end{proof}

\begin{remark}
	By using the method based on Formula \eqref{eq:cdf-alzaatreh}, the fact that $f_{R18a}(x)$ is a well defined pdf is equivalent to the fact that the function 
	\begin{equation} 
		\label{eq:r-r18a}
		r(t) = (1 + \lambda_1) + 2(\lambda_2 - \lambda_1) t - 3\lambda_2 t^2
	\end{equation}
	is positive and thus a pdf on $[0,1]$. Figure \ref{fig:verif-r18a} presents a graphical study (for $\lambda_1 \in [-1.5,3.5]$ and $\lambda_2 \in [-4.5,2.5]$ with discretization step equal to 0.1 for both parameters) of whether or not, $r(t)$ defined by Equation \eqref{eq:r-r18a} is a positive function on $[0,1]$. In this figure, the range \eqref{eq:range-r18a} proposed by \citet{rahman-et-al-2018a} is represented by a red dotted triangular area while the modified range \eqref{eq:range-r18a-b} proposed in this paper is represented by the black triangular area. We also notice that there exist values of $\lambda_1$ and $\lambda_2$ outside the conditions \eqref{eq:range-r18a-b} for which $r(t)$ is positive but the shape of the graph does not suggest any simple expression of the constraints on these parameters apart from the conditions \eqref{eq:range-r18a-b}.
	
	\begin{figure}[!h]
		\centering
		\includegraphics[scale=0.5]{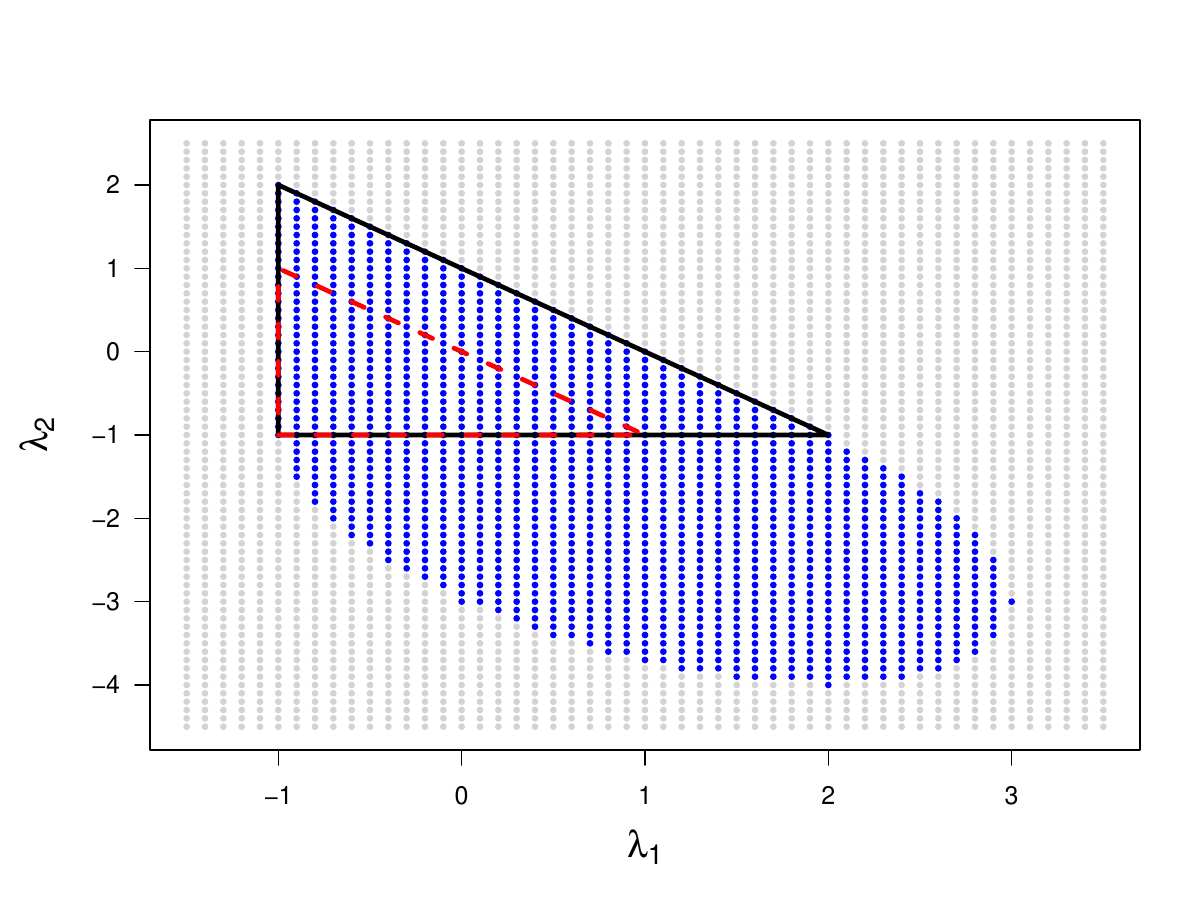}
		\caption{Graphical study for different values of $\lambda_1$ and $\lambda_2$ of whether or not $r(t) = (1 + \lambda_1) + 2(\lambda_2 - \lambda_1) t - 3\lambda_2 t^2$ is positive on $[0,1]$ (blue colour means "yes" and gray means "no").}
		\label{fig:verif-r18a}
	\end{figure}
\end{remark}

\begin{proposition}
	\label{prop:ra18b}
	Let $G(x)$ be the cdf of a random variable. 
	\begin{enumerate}[i)]
		\item The cubic transmuted family of distributions of parameters $\lambda_1$ and $\lambda_2$ proposed by \citet{rahman-et-al-2018b} (denoted $CT_{R18b}(\lambda_1,\lambda_2)$) linked to $G(x)$ is well defined under conditions \eqref{eq:range-r18b}. 
		\item The $CT_{R18b}(\lambda_1,\lambda_2)$ family is also well defined under the extended range
		\begin{equation} 
			\label{eq:range-r18b-b}
			\mathscr{S}_{MR18b} = \Big\{ (\lambda_1,\lambda_2) : -2 \leqslant \lambda_1 \leqslant 1, \; -2 \leqslant \lambda_2 \leqslant 1 \; \text{and} \; -1 \leqslant \lambda_1 + \lambda_2 \leqslant 2 \Big\}
		\end{equation}
		and will be called modified $CT_{R18b}$ and denoted $CT_{MR18b}$.
		\item For all $(\lambda_1,\lambda_2) \in \mathscr{S}_{MR18b}$, the $CT_{MR18b}(\lambda_1,\lambda_2)$ family of distributions linked to $G(x)$ is equivalent to the $CT_{MG}(1 + \lambda_1 + \lambda_2, 1-\lambda_2)$ family of distributions linked to $G(x)$.
		\item For all $(\lambda_1,\lambda_2) \in \mathscr{S}_{MG}$, the $CT_{MG}(\lambda_1,\lambda_2)$ family of distributions linked to $G(x)$ is equivalent to the $CT_{MR18b}(\lambda_1+\lambda_2-2, 1-\lambda_2)$ family of distributions linked to $G(x)$.
	\end{enumerate}
\end{proposition}

\begin{proof}
	\mbox{}
	\begin{enumerate}[i)]
		\item[ii)] By definition of the cdf of order statistics, the cdf from Equation \eqref{eq:cdf-granz-0} is well defined under the conditions $0 \leqslant \pi_i \leqslant 1$ for all $i=1,2,3$ and $\pi_1 + \pi_2 = 1-\pi_3 \in [0,1]$. Since $\lambda_1 = 1 - 3\pi_1$ and $\lambda_2= 1 - 3\pi_2$, we easily deduce that $-2 \leqslant \lambda_1 \leqslant 1$ and $-2 \leqslant \lambda_2 \leqslant 1$. We also have
		$$
		\lambda_1 + \lambda_2 = 2 - 3(\pi_1 + \pi_2) = 2 - 3(1 - \pi_3) = 3\pi_3 - 1
		$$
		hence $-1 \leqslant \lambda_1 + \lambda_2 \leqslant 2$. 
		
		\item[i)] It is immediate from ii) since the conditions \eqref{eq:range-r18b} set by \citet{rahman-et-al-2018b} on $\lambda_1$ and $\lambda_2$ are included in \eqref{eq:range-r18b-b}.
		
		\item[iii)] Let $(\lambda_1,\lambda_2) \in \mathscr{S}_{MR18b}$. It is easily seen that by replacing $\lambda_1$ by $\lambda_1' = 1 + \lambda_1 + \lambda_2$ and $\lambda_2$ by $\lambda_2' = 1 - \lambda_2$ in \eqref{eq:cdf-granz}, we obtain \eqref{eq:cdf-r18b}. We also easily verify that $(\lambda_1,\lambda_2) \in \mathscr{S}_{MR18b}$ if and only if $(1+\lambda_1+\lambda_2, 1-\lambda_2) \in \mathscr{S}_{MG}$. Indeed
		\begin{align*}
			(\lambda_1,\lambda_2) \in \mathscr{S}_{MR18b} & \iff \left\{ \begin{array}{c}
				-1 \leqslant \lambda_1 + \lambda_2 \leqslant 2 \\ -2 \leqslant \lambda_2 \leqslant 1 \\ -2 \leqslant \lambda_1 \leqslant 1
			\end{array} \right. \\
			& \iff \left\{ \begin{array}{c}
				0 \leqslant 1 + \lambda_1 + \lambda_2 \leqslant 3 \\ 0 \leqslant 1-\lambda_2 \leqslant 3 \\ 0 \leqslant 2+\lambda_1 \leqslant 3 
			\end{array} \right. \\
			& \iff (1+\lambda_1+\lambda_2, 1-\lambda_2) \in \mathscr{S}_{MG}
		\end{align*}
		which completes the proof because $(1+\lambda_1+\lambda_2) + (1-\lambda_2) = 2+\lambda_1$. 
		
		\item[iv)] Let $(\lambda_1,\lambda_2) \in \mathscr{S}_{MG}$. By replacing $\lambda_1$ by $\lambda_1' = \lambda_1 + \lambda_2 - 2$ and $\lambda_2$ by $\lambda_2' = 1 - \lambda_2$ in \eqref{eq:cdf-r18b}, we obtain \eqref{eq:cdf-granz}. We also easily verify that $(\lambda_1,\lambda_2) \in \mathscr{S}_{MG}$ if and only if $(\lambda_1+\lambda_2-2, 1-\lambda_2) \in \mathscr{S}_{MR18b}$. Indeed
		\begin{align*}
			(\lambda_1,\lambda_2) \in \mathscr{S}_{MG} & \iff \left\{ \begin{array}{c}
				0 \leqslant \lambda_1 + \lambda_2 \leqslant 3 \\ 0 \leqslant \lambda_2 \leqslant 3 \\ 0 \leqslant \lambda_1 \leqslant 3
			\end{array} \right. \\
			& \iff \left\{ \begin{array}{c}
				-2 \leqslant \lambda_1 + \lambda_2 - 2 \leqslant 1 \\ -2 \leqslant 1-\lambda_2 \leqslant 1 \\ -1 \leqslant \lambda_1 - 1 \leqslant 2
			\end{array} \right. \\
			& \iff (\lambda_1 + \lambda_2 - 2, 1-\lambda_2) \in \mathscr{S}_{MR18b}
		\end{align*}
		which completes the proof because $(\lambda_1 + \lambda_2 - 2) + (1-\lambda_2) = \lambda_1 - 1$. 
	\end{enumerate}
\end{proof}

\begin{remark}
	By using the method based on Formula \eqref{eq:cdf-alzaatreh}, the fact that $f_{R18b}(x)$ is a well defined pdf is equivalent to the fact that the function 
	\begin{equation} 
		\label{eq:r-r18b}
		r(t) = (1 + \lambda_1 + \lambda_2) - 2(\lambda_1 + 2\lambda_2) t + 3\lambda_2 t^2
	\end{equation}
	is positive and thus a pdf on $[0,1]$. Figure \ref{fig:verif-r18b} presents a graphical study (for $\lambda_1 \in[-3.5,1.5]$ and $\lambda_2 \in [-2.5,4.5]$ with discretization step equal to 0.1 for both parameters) of whether or not, $r(t)$ defined by Equation \eqref{eq:r-r18b} is a positive function on $[0,1]$. In this figure, the range \eqref{eq:range-r18b} proposed by \citet{rahman-et-al-2018b} is represented by a red dotted rectangular area while the modified range \eqref{eq:range-r18b-b} proposed in this paper is represented by the black triangular area. We also notice that there exist values of $\lambda_1$ and $\lambda_2$ outside the conditions \eqref{eq:range-r18b-b} for which $r(t)$ is positive but the shape of the graph does not suggest any simple expression of the constraints on these parameters apart from the conditions \eqref{eq:range-r18b-b}.
	
	\begin{figure}[!h]
		\centering
		\includegraphics[scale=0.5]{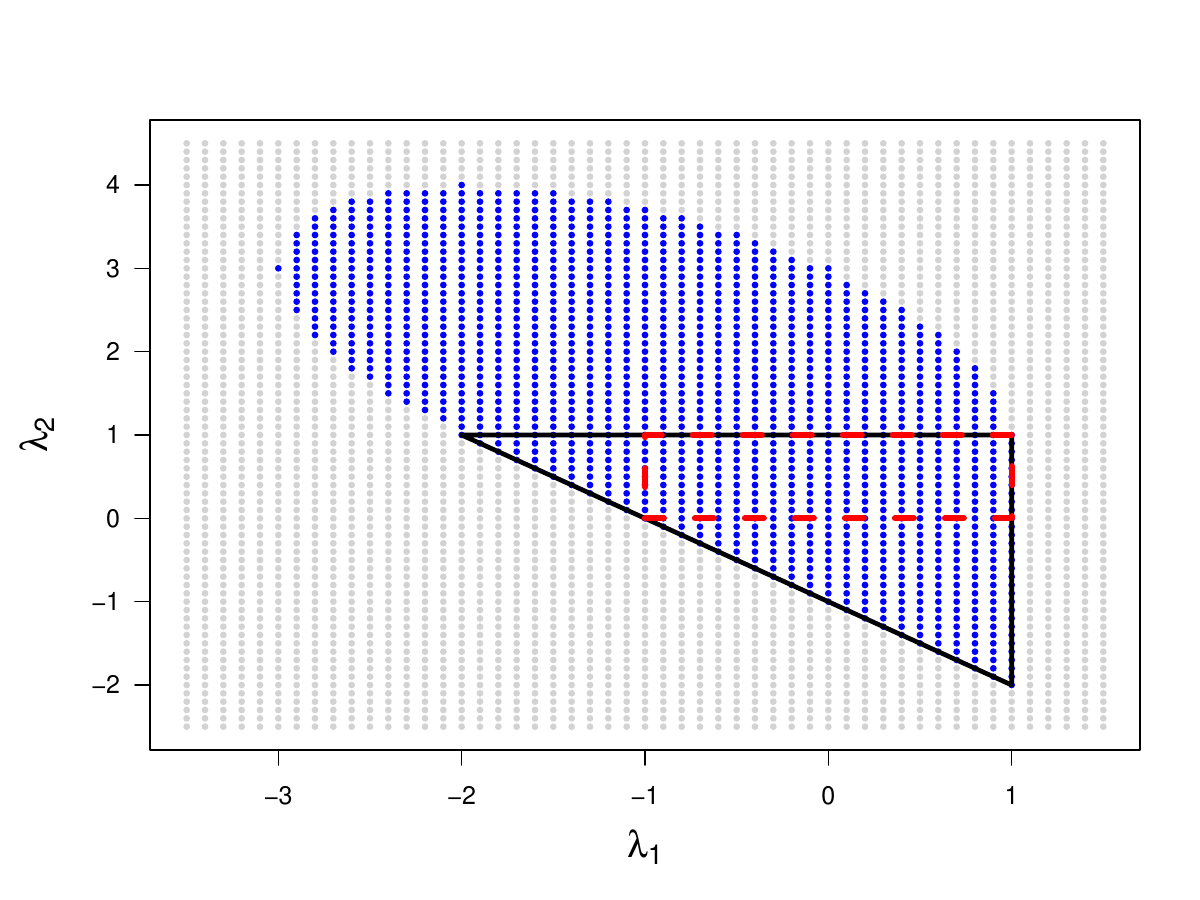}
		\caption{Graphical study for different values of $\lambda_1$ and $\lambda_2$ of whether or not $r(t) = (1 + \lambda_1 + \lambda_2) - 2(\lambda_1 + 2\lambda_2) t + 3\lambda_2 t^2$ is positive on $[0,1]$ (blue colour means "yes" and gray means "no").}
		\label{fig:verif-r18b}
	\end{figure}
\end{remark}

\begin{proposition}
	\label{prop:ra19}
	\mbox{}
	\begin{enumerate}[i)]
		\item The $CT_{R19}(\lambda)$ family of distributions defined by Equation \eqref{eq:cdf-r19} is not a sub-family of the $CT_G(\lambda_1,\lambda_2)$ family of distributions if $\lambda \neq 0$.
		\item It is rather a sub-family of the $CT_{MG}(\lambda_1,\lambda_2)$ family of distributions for $-\frac{1}{2} \leqslant \lambda \leqslant 1$. 
		\item Moreover, the $CT_{R19}(\lambda)$ distribution is well defined under the extended condition $-2 \leqslant \lambda \leqslant 1$ and will be called modified cubic Transmutation of \citet{rahman-et-al-2019b} and denoted $CT_{MR19}$.
	\end{enumerate}
\end{proposition}

\begin{proof} 
	\mbox{}
	\begin{enumerate}[i)]
		\item Let us first prove by absurd that the $CT_{R19}(\lambda)$ family of distributions is not a sub-family of the $CT_G(\lambda_1,\lambda_2)$ family if $\lambda \neq 0$. Let $\lambda \in [-1,1]$. Equation \eqref{eq:cdf-r19} can be rewritten under the form of Equation \eqref{eq:cdf-granz} by setting $\lambda_1 = 1-\lambda$ and $\lambda_2 = 1 + 2\lambda$. We would have 
		\begin{align*} 
			(\lambda_1, \lambda_2) = (1-\lambda, 1 + 2\lambda) \in \mathscr{S}_G & \iff \left\{ \begin{array}{c} 
				0 \leqslant 1-\lambda \leqslant 1 \\
				-1 \leqslant 1 + 2\lambda \leqslant 1 
			\end{array} \right. \\
			& \iff \left\{ \begin{array}{c} 
				0 \leqslant \lambda \leqslant 1 \\
				-1 \leqslant \lambda \leqslant 0 
			\end{array} \right. \\
			& \iff \lambda=0.
		\end{align*}
		Thus, the $CT_{R19}(\lambda)$ family of distributions is a sub-family of the $CT_G(\lambda_1,\lambda_2)$ family only if $\lambda = 0$.
		
		\item Let $\lambda \in [-1,1]$, $\lambda_1 = 1 - \lambda$ and $\lambda_2 = 1 + 2\lambda$. 
		\begin{align*} 
			\left. \begin{array}{c} 
				0 \leqslant \lambda_1 \leqslant 3 \\ 
				0 \leqslant \lambda_2 \leqslant 3 \\
				0 \leqslant \lambda_1 + \lambda_2 \leqslant 3
			\end{array} \right\} & \iff \left\{ \begin{array}{c} 
				0 \leqslant 1 - \lambda \leqslant 3 \\ 
				0 \leqslant 1 + 2\lambda \leqslant 3  \\
				0 \leqslant \lambda + 2 \leqslant 3
			\end{array} \right. \\
			& \iff \left\{ \begin{array}{c} 
				-2 \leqslant \lambda \leqslant 1 \\ 
				-\frac{1}{2} \leqslant \lambda \leqslant 1 \\
				-2 \leqslant \lambda \leqslant 1
			\end{array} \right. \\
			& \iff -\frac{1}{2} \leqslant \lambda \leqslant 1.
		\end{align*}
		
		\item Let us now prove that the $CT_{R19}$ distribution is well defined under the extended condition $-2 \leqslant \lambda \leqslant 1$. The pdf \eqref{eq:pdf-r19} is given by 
		$$
		f_{R19}(x) = g(x) \, r(G(x)),
		$$
		where for all $t \in [0,1]$, $r(t) = (1 - \lambda) + 6\lambda t - 6\lambda t^2$ and, according to \cite{alzaatreh-et-al-2013}, $f_{R19}(x)$ is well defined if $r(t)$ is a pdf on $[0,1]$. We have
		$$
		\int_{0}^{1} r(t) \dt = 1.
		$$
		So for $r(t)$ to be a pdf, it is enough to check that for all $t \in [0,1]$, $r(t) \geqslant 0$. We have $r'(t) = -6\lambda (2t-1)$. 
		\begin{list}{$\circ$}{}
			\item If $\lambda \leqslant 0$, then the variation table of $r(t)$ is as follows:
			$$
			\begin{tabvar}{|C|LCCCC|} 
				\hline t & 0 & & \frac{1}{2} & & 1 \\ 
				\hline r'(t) & & - & \barre{0} & + & \\ 
				\hline \niveau{2}{2} \TVcenter{r(t)} & \TVstretch[3pt]{1-\lambda} & \decroit & \TVstretch[3pt]{\frac{\lambda + 2}{2}} & \croit & 1-\lambda \\ 
				\hline		
			\end{tabvar}
			$$
			It is enough that $\frac{\lambda + 2}{2} \geqslant 0$ i.e. $\lambda \geqslant -2$.
			
			\item If $\lambda \geqslant 0$, then the variation table of $r(t)$ is as follows: 
			$$
			\begin{tabvar}{|C|LCCCC|} 
				\hline t & 0 & & \frac{1}{2} & & 1 \\ 
				\hline r'(t) & & + & \barre{0} & - & \\ 
				\hline \niveau{1}{2} \TVcenter{r(t)} & 1-\lambda & \croit & \TVstretch[3pt]{\frac{\lambda + 2}{2}} & \decroit & 1-\lambda \\ 
				\hline		
			\end{tabvar}
			$$
			It is enough that $1-\lambda \geqslant 0$ i.e. $\lambda \leqslant 1$.
		\end{list}
		Thus, $r(t)$ is positive on $[0,1]$ if $-2 \leqslant \lambda \leqslant 0$ or $0 \leqslant \lambda \leqslant 1$. We conclude that $r(t)$ is positive on $[0,1]$ if $-2 \leqslant \lambda \leqslant 1$.
	\end{enumerate}	
\end{proof}

\begin{remark}
	By using the method based on Formula \eqref{eq:cdf-alzaatreh}, the fact that $f_{R23}(x)$ is a well defined pdf is equivalent to the fact that the function 
	\begin{equation} 
		\label{eq:r-r23}
		r(t) = 1 - \lambda (\eta - 1) + 2\lambda (2\eta - 1) t - 3 \lambda \eta t^2
	\end{equation}
	is positive and thus a pdf on $[0,1]$. As suggested by some comments and calculations in \citep{hameldarbandi-yilmaz-2020,yilmaz-2025}, the theoretical determination of the range is far from easy when the cubic transmutation contains two parameters. We can then resort to a graphic study. Figure \ref{fig:verif-r23} presents a graphical study (for $\lambda \in [-3.2,1.2]$ and $\eta \in [-3,3]$ with discretization step equal to 0.1 for both parameters) of whether or not $r(t)$ defined by Equation \eqref{eq:r-r23} is a positive function on $[0,1]$. In this figure, the range \eqref{eq:range-r23} proposed by \citet{rahman-et-al-2023} is represented by a rectangular area in red dashes. We also notice that there exist values of $\lambda$ and $\eta$ outside the conditions \eqref{eq:range-r23} for which $r(t)$ is positive but the shape of the graph does not suggest any simple expression of the constraints on these parameters apart from the conditions \eqref{eq:range-r23}.
	
	\begin{figure}[!h]
		\centering
		\includegraphics[scale=0.5]{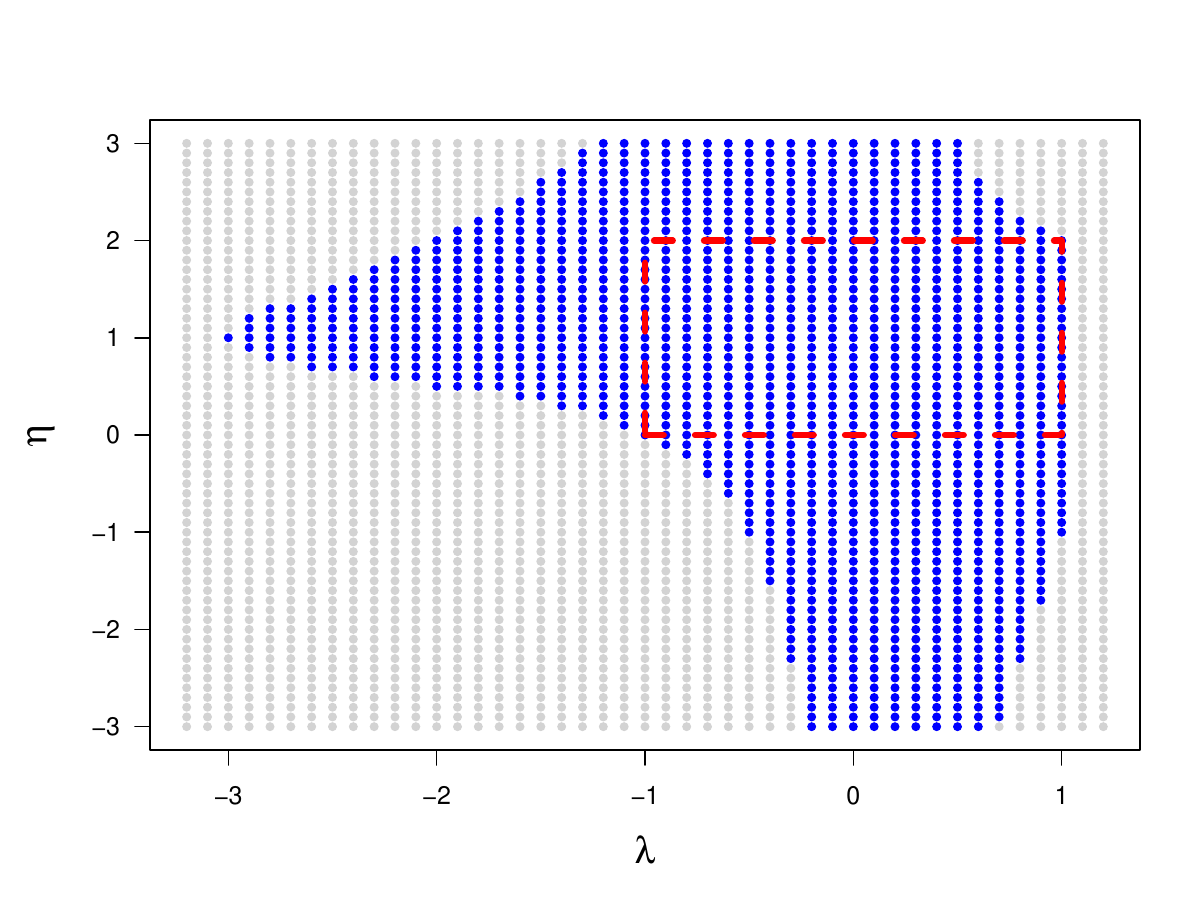}
		\caption{Graphical study for different values of $\lambda$ and $\eta$ of whether or not $r(t) = (1 + \lambda(1 - \eta)) + 2\lambda(2\eta - 1) t - 3\lambda \eta t^2$ is positive on $[0,1]$ (blue colour means "yes" and gray means "no").}
		\label{fig:verif-r23}
	\end{figure}
\end{remark}

\section{Illustration of the results using Pareto distribution}
\label{sec:illustration}

In this section, we illustrate our results using Pareto distribution as baseline distribution. We compare the Pareto distribution, transmuted Pareto (TP) distribution, the cubic transmuted Pareto (CTP) distributions and their modifications on real datasets using R software \citep{R-2024}. For this purpose, we tested several datasets including those which have been used in the past for the quadratic and cubic transmutations of Pareto distribution in \citep{merovci-puka-2014,ansari-eledum-2018,rahman-et-al-2020a,rahman-et-al-2021}. The datasets selected and presented hereafter are mainly those which illustrate the theoretical results presented earlier in this paper.

For each dataset, the selection of the best model (the distribution which best fits each dataset) is done using the classical likelihood-based model selection criteria which are the negative of the log-likelihood ($-\log L$), the Akaike Information Criterion (AIC), the Akaike Information Criterion Corrected (AICC), the Bayesian Information Criterion (BIC) respectively defined as
\begin{equation*} 
	\text{AIC} = -2 \log L^{*} + 2k, \quad \text{AICC} = \text{AIC} + \frac{2k(k+1)}{n-(k+1)}, \quad \text{BIC} = -2 \log L^{*} + k \log n,
\end{equation*}
where $k$ is the number of parameters to estimate, $n$ is the number of observations in the dataset and $\log L^{*}$ is the maximum value of the log-likelihood. For each dataset, the best model is the one with the smallest values for the criteria. 

We first present the pdfs, likelihoods and maximum likelihood estimators (MLE) linked to each CTP.

\subsection{Probability density functions and MLE}

\subsubsection{Probability density functions}

Pareto distribution is defined by its cdf
\begin{equation}
	G(x) = 1 - \left(\frac{x_0}{x} \right)^{\alpha}, \quad x \geqslant x_0,
\end{equation}
and its pdf
\begin{equation}
	g(x) = \frac{\alpha x_0^{\alpha}}{x^{\alpha+1}}, \quad x \geqslant x_0,
\end{equation}
where $x_0 \in \mathbb{R}_{+}^{*}$ and $\alpha \in \mathbb{R}_{+}^{*}$. Using the quadratic transmutation proposed by \citet{shaw-buckley-2009}, \citet{merovci-puka-2014} developed the transmuted Pareto (TP) distribution. 

Several authors have been interested in cubic transmuted Pareto distribution (CTP). Using the CT formula of \citet{alkadim-mohammed-2017}, \citet{ansari-eledum-2018} proposed a CTP distribution that we will denote $CTP_A$ and whose pdf is 
\begin{equation}
	f_A(x) = \frac{\alpha x_0^{\alpha}}{x^{\alpha + 1}} \left[ 1 - 2\lambda \left( \frac{x_0}{x} \right)^{\alpha} + 3\lambda \left( \frac{x_0}{x} \right)^{2\alpha} \right].
\end{equation}
\citet{rahman-et-al-2020a} have used the formula of \citet{rahman-et-al-2018a} and proposed a CTP distribution (that we will denote $CTP_{R18a}$) with corresponding pdf
\begin{equation}
	f_{R18a}(x) = \frac{\alpha x_0^{\alpha}}{x^{\alpha + 1}} \left[ (1 - \lambda_1 - \lambda_2) + 2(\lambda_1 + 2\lambda_2) \left( \frac{x_0}{x} \right)^{\alpha} - 3\lambda_2 \left( \frac{x_0}{x} \right)^{2\alpha} \right].
\end{equation}
\citet{rahman-et-al-2021} proposed another CTP distribution (that we will denote $CTP_{R18b}$) using the formula of \citet{rahman-et-al-2018b} and the corresponding pdf is
\begin{equation}
	f_{R18b}(x) = \frac{\alpha x_0^{\alpha}}{x^{\alpha + 1}} \left[ (1 - \lambda_1) + 2(\lambda_1 - \lambda_2) \left( \frac{x_0}{x} \right)^{\alpha} + 3\lambda_2 \left( \frac{x_0}{x} \right)^{2\alpha} \right].
\end{equation}
To the best of our knowledge, the other formulas of CT have not been applied so far to the Pareto distribution. The CTP linked to the formulas of \Citet{granzotto-et-al-2017}, \citet{rahman-et-al-2019b} and \citet{rahman-et-al-2023} will be respectively denoted $CTP_G$, $CTP_{R19}$ and $CTP_{R23}$ and their respective pdf are
\begin{equation}
	f_G(x) = \frac{\alpha x_0^{\alpha}}{x^{\alpha + 1}} \left[ (3 - \lambda_1 - \lambda_2) + 2(\lambda_1 + 2\lambda_2 - 3) \left( \frac{x_0}{x} \right)^{\alpha} + 3(1 - \lambda_2) \left( \frac{x_0}{x} \right)^{2\alpha} \right],
\end{equation}
\begin{equation}
	f_{R19}(x) = \frac{\alpha x_0^{\alpha}}{x^{\alpha + 1}} \left[ (1-\lambda) + 6\lambda \left( \frac{x_0}{x} \right)^{\alpha} - 6\lambda \left( \frac{x_0}{x} \right)^{2\alpha} \right]
\end{equation}
and
\begin{equation}
	f_{R23}(x) = \frac{\alpha x_0^{\alpha}}{x^{\alpha + 1}} \left[ (1 - \lambda) + 2\lambda (1 + \eta) \left( \frac{x_0}{x} \right)^{\alpha} - 3 \lambda \eta \left( \frac{x_0}{x} \right)^{2\alpha} \right].
\end{equation}
The modified versions of $CTP_{A}$, $CTP_{R18a}$, $CTP_{R18b}$, $CTP_{G}$ and $CTP_{R19}$ will be respectively denoted $CTP_{MA}$, $CTP_{MR18a}$, $CTP_{MR18b}$, $CTP_{MG}$ and $CTP_{MR19}$.

\subsubsection{Estimation of parameters using the maximum likelihood method}

Let $x_1, \ldots, x_n$ be a random sample of size $n$ from a given CTP distribution. The general form of the log-likelihood is
\begin{equation}
	\log L = \sum_{i=1}^n \log f(x_i),
\end{equation}
where $f$ is the pdf. 

Since for all $i=1,\ldots,n$, $x_0 \leqslant x_i$, the maximum likelihood estimator (MLE) of $x_0$ is the minimum of the sample values i.e. $\min_{1 \leqslant i \leqslant n} x_i$. The MLEs of the other parameters are obtained by maximizing the respective log-likelihoods of the different models defined below.

\begin{list}{$\bullet$}{}
	\item The log-likelihood corresponding to both the $CTP_G$ and the $CTP_{MG}$ is
	\begin{multline} 
		\log L_G(\alpha, \lambda_1, \lambda_2) = n \log(\alpha) + n\alpha \log(x_0) - (\alpha+1) \sum_{i=1}^n \log(x_i) \\
		+ \sum_{i=1}^n \log \left[ (3 - \lambda_1 - \lambda_2) + (2\lambda_1 + 4\lambda_2 - 6) \left( \frac{x_0}{x_i} \right)^{\alpha} \right. \\
		\left. + (3 - 3\lambda_2) \left( \frac{x_0}{x_i} \right)^{2\alpha} \right].
	\end{multline}
	The maximum likelihood estimates (MLE) $\hat\alpha$, $\hat\lambda_1$ and $\hat\lambda_2$ of $\alpha$, $\lambda_1$ and $\lambda_2$ are such that
	\begin{equation}
		(\hat\alpha, \hat\lambda_1, \hat\lambda_2) = \argmax_{(\alpha, \lambda_1, \lambda_2) \in \reels_+^* \times \mathscr{S}} \log L_G(\alpha, \lambda_1, \lambda_2),
	\end{equation}
	where $\mathscr{S} = \mathscr{S}_{G}$ if we consider the $CTP_G$ and $\mathscr{S} = \mathscr{S}_{MG}$ if we consider the $CTP_{MG}$.
	
	\item The log-likelihood corresponding to both the $CTP_A$ of \citet{ansari-eledum-2018} and the $CTP_{MA}$ is
	\begin{multline}
		\log L_A(\alpha, \lambda) = n \log(\alpha) + n\alpha \log(x_0) - (\alpha+1) \sum_{i=1}^n \log(x_i) \\
		+ \sum_{i=1}^n \log \left[ 1 - 2\lambda \left( \frac{x_0}{x_i} \right)^{\alpha} + 3\lambda \left( \frac{x_0}{x_i} \right)^{2\alpha} \right].
	\end{multline}
	The MLE $\hat\alpha$ and $\hat\lambda$ of $\alpha$ and $\lambda$ are such that
	\begin{equation}
		(\hat\alpha, \hat\lambda) = \argmax_{(\alpha, \lambda) \in \reels_+^* \times \mathscr{S}} \log L_A(\alpha, \lambda),
	\end{equation}
	where $\mathscr{S} = [-1,1]$ if we consider the $CTP_A$ and $\mathscr{S} = [-1,3]$ if we consider the $CTP_{MA}$.
	
	\item The log-likelihood corresponding to both $CTP_{R18a}$ \citep{rahman-et-al-2020a} and $CTP_{MR18a}$ is
	\begin{multline}
		\log L_{R18a}(\alpha, \lambda_1, \lambda_2) = n \log(\alpha) + n\alpha \log(x_0) - (\alpha+1) \sum_{i=1}^n \log(x_i) \\
		+ \sum_{i=1}^n \log \left[ (1 - \lambda_1 - \lambda_2) + 2(\lambda_1 + 2\lambda_2) \left( \frac{x_0}{x_i} \right)^{\alpha} - 3\lambda_2 \left( \frac{x_0}{x_i} \right)^{2\alpha} \right].
	\end{multline}
	The MLE $\hat\alpha$, $\hat\lambda_1$ and $\hat\lambda_2$ of $\alpha$, $\lambda_1$ and $\lambda_2$ are such that
	\begin{equation}
		(\hat\alpha, \hat\lambda_1, \hat\lambda_2) = \argmax_{(\alpha, \lambda_1, \lambda_2) \in \reels_+^* \times \mathscr{S}} \log L_{R18a}(\alpha, \lambda_1, \lambda_2),
	\end{equation}
	where $\mathscr{S} = \mathscr{S}_{R18a}$ if we consider the $CTP_{R18a}$ and $\mathscr{S} = \mathscr{S}_{MR18a}$ if we consider the $CTP_{MR18a}$.
	
	\item The log-likelihood corresponding to both $CTP_{R18b}$ of \citet{rahman-et-al-2021} and $CTP_{MR18b}$ is 
	\begin{multline}
		\log L_{R18b}(\alpha, \lambda_1, \lambda_2) = n \log(\alpha) + n\alpha \log(x_0) - (\alpha+1) \sum_{i=1}^n \log(x_i) \\
		+ \sum_{i=1}^n \log \left[ (1 - \lambda_1) + 2(\lambda_1 - \lambda_2) \left( \frac{x_0}{x_i} \right)^{\alpha} + 3\lambda_2 \left( \frac{x_0}{x_i} \right)^{2\alpha} \right].
	\end{multline}
	The MLE $\hat\alpha$, $\hat\lambda_1$ and $\hat\lambda_2$ of $\alpha$, $\lambda_1$ and $\lambda_2$ are such that
	\begin{equation}
		(\hat\alpha, \hat\lambda_1, \hat\lambda_2) = \argmax_{(\alpha, \lambda_1, \lambda_2) \in \reels_+^* \times \mathscr{S}} \log L_{R18b}(\alpha, \lambda_1, \lambda_2),
	\end{equation}
	where $\mathscr{S} = \mathscr{S}_{R18b}$ if we consider the $CTP_{R18b}$ and $\mathscr{S} = \mathscr{S}_{MR18b}$ if we consider the $CTP_{MR18b}$.
	
	\item The log-likelihood corresponding to both $CTP_{R19}$ and $CTP_{MR19}$ is
	\begin{multline}
		\log L_{R19}(\alpha, \lambda) = n \log(\alpha) + n\alpha \log(x_0) - (\alpha+1) \sum_{i=1}^n \log(x_i) \\
		+ \sum_{i=1}^n \log \left[ (1-\lambda) + 6\lambda \left( \frac{x_0}{x_i} \right)^{\alpha} - 6\lambda \left( \frac{x_0}{x_i} \right)^{2\alpha} \right].
	\end{multline}
	The MLE $\hat\alpha$ and $\hat\lambda$ of $\alpha$ and $\lambda$ are such that
	\begin{equation}
		(\hat\alpha, \hat\lambda) = \argmax_{(\alpha, \lambda) \in \reels_+^* \times \mathscr{S}} \log L_{R19}(\alpha, \lambda),
	\end{equation}
	where $\mathscr{S} = [-1,1]$ if we consider the $CTP_{R19}$ and $\mathscr{S} = [-2,1]$ if we consider the $CTP_{MR19}$.
	
	\item The log-likelihood corresponding to $CTP_{R23}$ is
	\begin{multline}
		\log L_{R23}(\alpha, \lambda, \eta) = n \log(\alpha) + n\alpha \log(x_0) - (\alpha+1) \sum_{i=1}^n \log(x_i) \\
		+ \sum_{i=1}^n \log \left[ (1 - \lambda) + 2\lambda (1 + \eta) \left( \frac{x_0}{x_i} \right)^{\alpha} - 3 \lambda \eta \left( \frac{x_0}{x_i} \right)^{2\alpha} \right].
	\end{multline}
	The MLE $\hat\alpha$, $\hat\lambda$ and $\hat\eta$ of $\alpha$, $\lambda$ and $\eta$ are such that
	\begin{equation}
		(\hat\alpha, \hat\lambda, \hat\eta) = \argmax_{(\alpha, \lambda, \eta) \in \reels_+^* \times [-1,1] \times [0,2]} \log L_{R23}(\alpha, \lambda, \eta).
	\end{equation}
\end{list}

All these constrained optimization problems are difficult to solve analytically and therefore require the use of a numerical optimization algorithm. One of the requirements for such algorithms is that they take inequality constraints into account. For our real data analysis hereafter, we will use the R function \textbf{constrOptim} which contains optimization algorithms integrating inequality constraints. 

\subsection{Floyd river dataset}

The Floyd data concerns the flood rates of the Floyd River (USA) from the years 1935 to 1973 (see \citet{mudholkar-hutson-1996} for more details on these data). Descriptive statistics for the Floyd data are summarized in Table \ref{tab:floyd-stats}. 

\begin{table}[!h]
	\centering
	\caption{Descriptive statistics for the Floyd dataset}
	\label{tab:floyd-stats}
	\begin{tabularx}{0.8\textwidth}{*{6}{Y}}
		\hline Min & $Q_1$ & Median & Mean & $Q_3$ & Max \\
		\hline 318 & 1590 & 3570 & 6771 & 6725 & 71500  \\ 
		\hline 
	\end{tabularx}
\end{table} 

\subsubsection{Fitting with the initial (unmodified) distributions}

Parameter estimates for the unmodified distributions are reported in Table \ref{tab:floyd-results-1} and Table \ref{tab:floyd-comp-1} gives the different criteria for the unmodified models.

\begin{table}[!h]
	\centering
	\caption{Parameters estimated for unmodified models on the Floyd dataset}
	\label{tab:floyd-results-1}
	\begin{tabularx}{0.8\textwidth}{cYYcY}
		\hline Distributions & \multicolumn{4}{c}{Estimations} \\ 
		\hline $CTP_{G}$ & $x_0 = 318$ & $\hat\alpha = 0.808$ & $\hat\lambda_1 = 0.104$ & $\hat\lambda_2 = -1$ \\ 
		$CTP_{A}$ & $x_0 = 318$ & $\hat\alpha = 0.436$ & $\hat\lambda = -0.876$ & \\                  
		$CTP_{R18a}$ & $x_0 = 318$ & $\hat\alpha = 0.718$ & $\hat\lambda_1 = -0.908$ & $\hat\lambda_2 = -1$ \\
		$CTP_{R18b}$ & $x_0 = 318$ & $\hat\alpha = 0.602$ & $\hat\lambda_1 = -1$ & $\hat\lambda_2 = 0.138$ \\
		$CTP_{R19}$ & $x_0 = 318$ & $\hat\alpha = 0.339$ & $\hat\lambda = 0.901$ & \\                  
		$CTP_{R23}$ & $x_0 = 318$ & $\hat\alpha = 0.33$ & $\hat\lambda = 1$ & $\hat\eta = 1.845$ \\          
		TP & $x_0 = 318$ & $\hat\alpha = 0.586$ & $\hat\lambda = -0.91$ & \\                   
		Pareto & $x_0 = 318$ & $\hat\alpha = 0.412$ & & \\
		\hline 
	\end{tabularx}
\end{table} 

\begin{table}[!h]
	\centering
	\caption{Criteria for comparing unmodified models on the Floyd dataset}
	\label{tab:floyd-comp-1}
	\begin{tabularx}{0.8\textwidth}{cYYYY}
		\hline Distributions & $-\log L^{*}$ & AIC & AICC & BIC  \\
		\hline  
		$CTP_{G}$    & $375.626^{(1)}$ & $757.252^{(1)}$ & $757.938^{(1)}$ & $762.243^{(1)}$ \\
		$CTP_{R18a}$ & $380.665^{(2)}$ & $767.330^{(2)}$ & $768.016^{(2)}$ & $772.321^{(2)}$ \\
		$CTP_{R23}$  & $384.719^{(3)}$ & $775.438^{(5)}$ & $776.124^{(5)}$ & $780.429^{(5)}$ \\
		$CTP_{R18b}$ & $384.915^{(4)}$ & $775.830^{(6)}$ & $776.516^{(6)}$ & $780.821^{(6)}$ \\
		$CTP_{R19}$  & $385.161^{(5)}$ & $774.322^{(3)}$ & $774.655^{(3)}$ & $777.649^{(3)}$ \\
		TP           & $385.349^{(6)}$ & $774.698^{(4)}$ & $775.031^{(4)}$ & $778.025^{(4)}$ \\
		$CTP_{A}$    & $387.052^{(7)}$ & $778.104^{(7)}$ & $778.437^{(7)}$ & $781.431^{(7)}$ \\
		Pareto       & $392.810^{(8)}$ & $787.620^{(8)}$ & $787.728^{(8)}$ & $789.284^{(8)}$ \\
		\hline
	\end{tabularx}
\end{table}

According to the results in Table \ref{tab:floyd-comp-1}, it seems that the distribution that best fits the Wheaton dataset is $CTP_G$. But in this case too, the $CTP_G$ distribution is not a true probability distribution in the rigorous sense of the term. Figure \ref{fig:floyd-granz} represents the fitted cdf (left) and pdf (right) from unmodified Granzotto cubic transmutation ($CTP_G$) for Floyd data for $x \in [318,850]$. We notice that the cdf is negative in the interval $[372.4243,640.6569]$ and the pdf is negative in the interval $[341.1465,505.9279]$.

\begin{figure}[!h]
	\centering
	\includegraphics[scale=0.35]{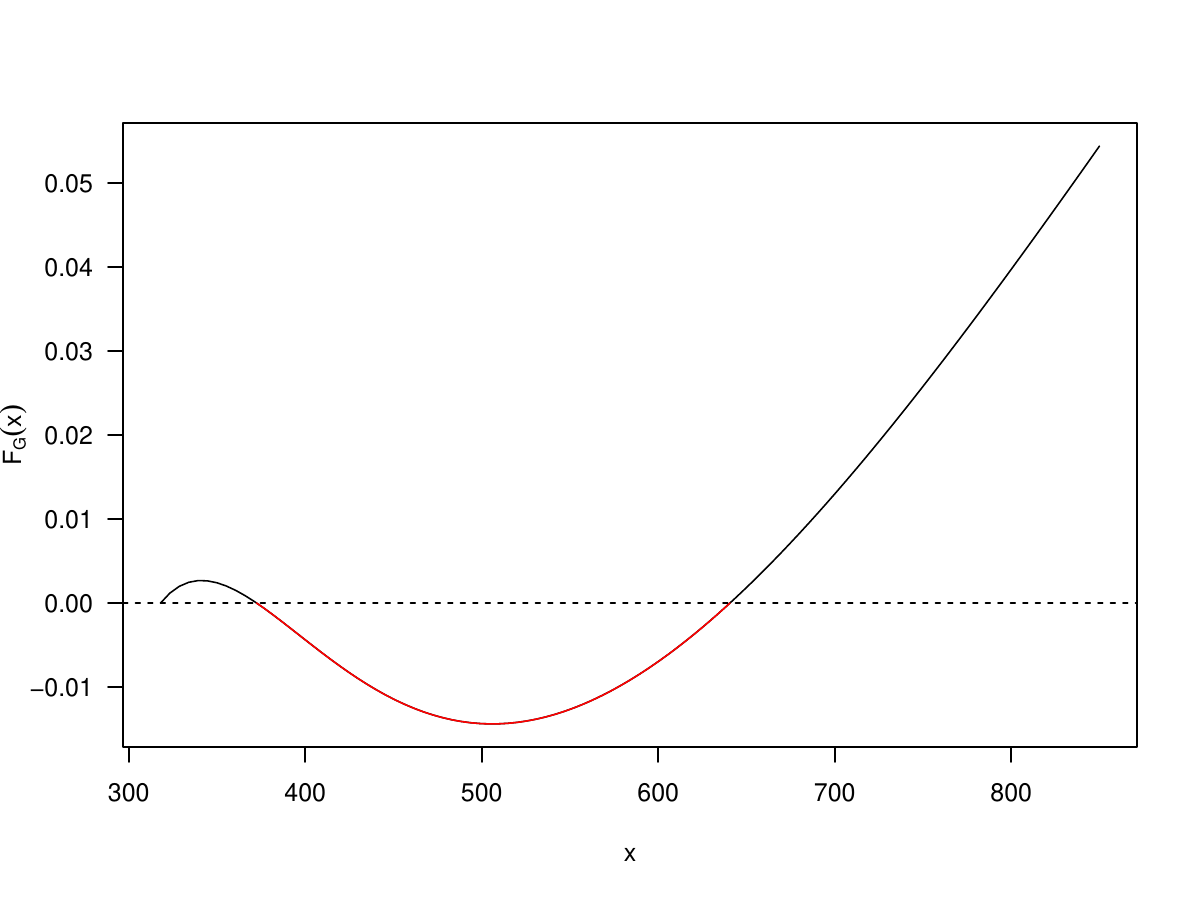}
	\includegraphics[scale=0.35]{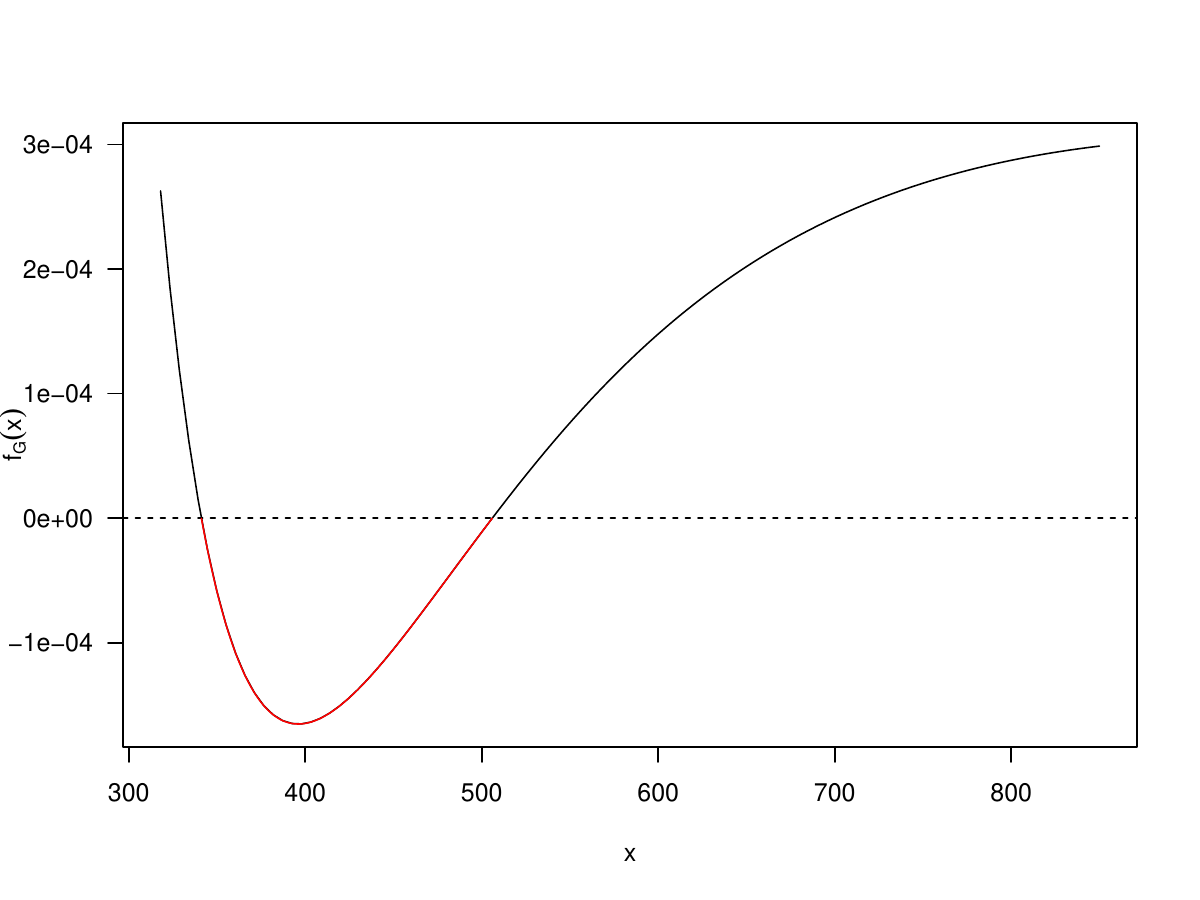}
	\caption{The fitted cdf (left) and pdf (right) from unmodified Granzotto cubic transmutation ($CTP_G$) for Floyd data}
	\label{fig:floyd-granz}
\end{figure}

\subsubsection{Fitting with modified distributions}

Table \ref{tab:floyd-results-2} summarizes the estimated parameter values for the modified models, while Table \ref{tab:floyd-comp-2} presents the model comparison criteria.

\begin{table}[!h]
	\centering
	\caption{Estimated parameters for models modified on the Floyd dataset}
	\label{tab:floyd-results-2}
	\begin{tabularx}{0.8\textwidth}{cYYcY}
		\hline Distributions & \multicolumn{4}{c}{Estimations} \\ 
		\hline 
		$CTP_{MG}$    & $x_0 = 318$ & $\hat\alpha = 0.719$ & $\hat\lambda_1 = 0.092$ & $\hat\lambda_2 = 0$ \\
		$CTP_{MA}$    & $x_0 = 318$ & $\hat\alpha = 0.436$ & $\hat\lambda = -0.876$ & \\   
		$CTP_{MR18a}$ & $x_0 = 318$ & $\hat\alpha = 0.719$ & $\hat\lambda_1 = -0.908$ & $\hat\lambda_2 = -1$ \\
		$CTP_{MR18b}$ & $x_0 = 318$ & $\hat\alpha = 0.718$ & $\hat\lambda_1 = -1.908$ & $\hat\lambda_2 = 1$ \\
		$CTP_{MR19}$  & $x_0 = 318$ & $\hat\alpha = 0.339$ & $\hat\lambda = 0.901$ & \\                
		$CTP_{R23}$   & $x_0 = 318$ & $\hat\alpha = 0.33$ & $\hat\lambda = 1$ & $\hat\eta = 1.845$ \\      
		TP            & $x_0 = 318$ & $\hat\alpha = 0.586$ & $\hat\lambda = -0.91$ & \\                  
		Pareto        & $x_0 = 318$ & $\hat\alpha = 0.412$ & & \\
		\hline 
	\end{tabularx}
\end{table} 

\begin{table}[!h]
	\centering
	\caption{Criteria for modified models on the Floyd dataset (ranks in parentheses)}
	\label{tab:floyd-comp-2}
	\begin{tabularx}{0.8\textwidth}{cYYYY}
		\hline Distributions & $-\log L^{*}$ & AIC & AICC & BIC  \\
		\hline  
		$CTP_{MG}$    & $380.665^{(1)}$ & $767.330^{(1)}$ & $768.016^{(1)}$ & $772.321^{(1)}$ \\
		$CTP_{MR18a}$ & $380.665^{(1)}$ & $767.330^{(1)}$ & $768.016^{(1)}$ & $772.321^{(1)}$ \\
		$CTP_{MR18b}$ & $380.665^{(1)}$ & $767.330^{(1)}$ & $768.016^{(1)}$ & $772.321^{(1)}$ \\
		$CTP_{R23}$   & $384.719^{(4)}$ & $775.438^{(6)}$ & $776.124^{(6)}$ & $780.429^{(6)}$ \\
		$CTP_{MR19}$  & $385.161^{(5)}$ & $774.322^{(4)}$ & $774.655^{(4)}$ & $777.649^{(4)}$ \\
		TP            & $385.349^{(6)}$ & $774.698^{(5)}$ & $775.031^{(5)}$ & $778.025^{(5)}$ \\
		$CTP_{MA}$    & $387.052^{(7)}$ & $778.104^{(7)}$ & $778.437^{(7)}$ & $781.431^{(7)}$ \\
		Pareto        & $392.810^{(8)}$ & $787.620^{(8)}$ & $787.728^{(8)}$ & $789.284^{(8)}$ \\
		\hline
	\end{tabularx}
\end{table}

A comparison between Tables \ref{tab:floyd-results-1} and \ref{tab:floyd-results-2} shows that the parameters of the $CTP_{G}$ and $CTP_{R18b}$ distributions have changed, with the best-fitting distributions now being $CTP_{MG}$, $CTP_{MR18a}$ and $CTP_{MR18b}$. The modifications introduced have improved the $CTP_{G}$, ensuring it is well defined as a probability distribution and have elevated the $CTP_{R18b}$ from the 4th place to 1st (tied).

Figure \ref{fig:floyd} displays the histogram of the Floyd dataset along with the fitted densities of the modified models.

\begin{figure}[!h]
	\centering
	\includegraphics[scale=0.5]{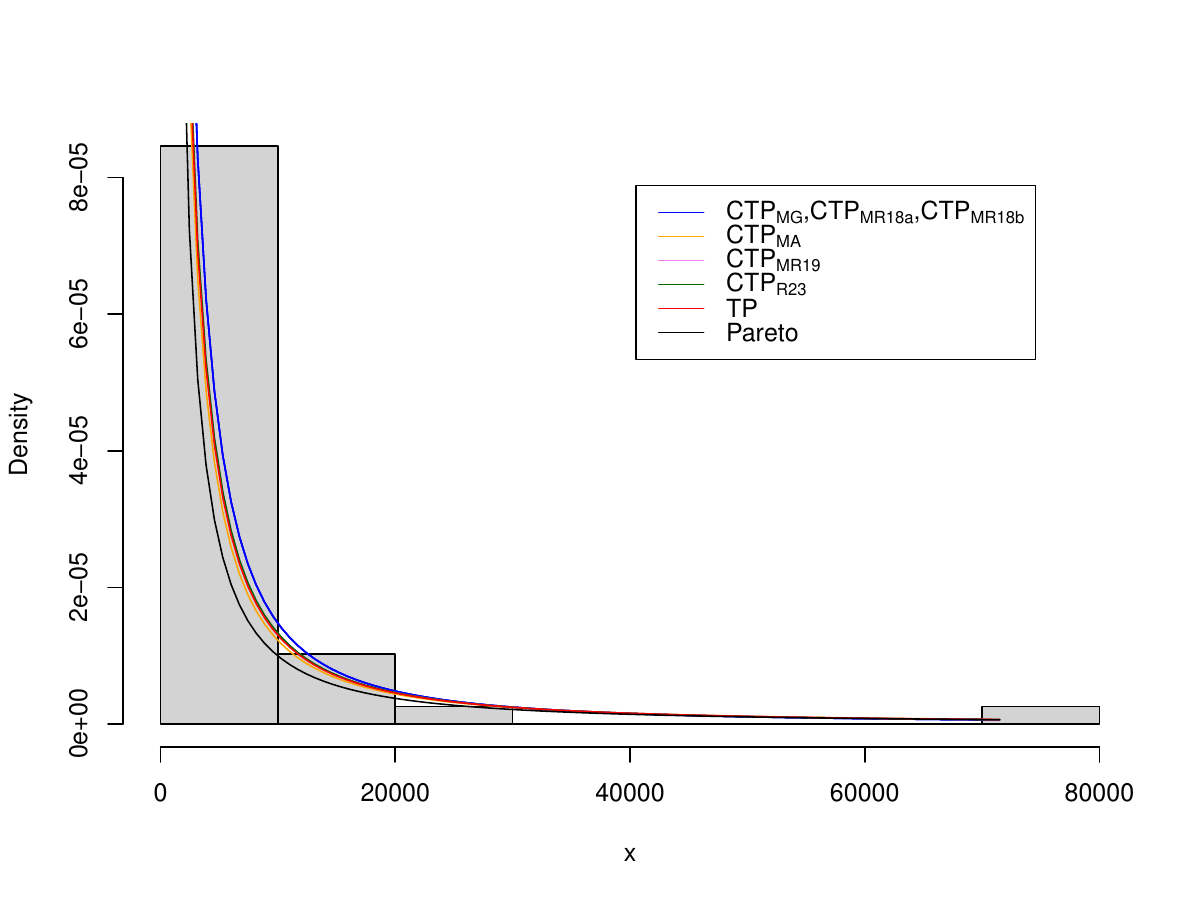}
	\caption{Graphical comparison of fitting densities of modified models on the Floyd dataset}
	\label{fig:floyd}
\end{figure}

\subsection{Population size of major urban agglomerations} 

The Pareto distribution has been applied by \citet{barranco-jimenez-2012} and \citet{bedbur-et-al-2019} to the modelling of the size of the hundred (100) largest urban agglomerations in the world at the respective reference dates 01/01/2010 and 01/01/2016. The data are provided by \citet{brinkhoff-2024} through the website \url{https://www.citypopulation.de/en/world/agglomerations/}. We will compare the different distributions considered in this paper for fitting the population size of the hundred (100) largest urban agglomerations in the world at the reference date 01/01/2024. The descriptive statistics for this dataset are reported in Table \ref{tab:pop-stats}.

\begin{table}[!h]
	\centering
	\caption{Descriptive statistics for Population size data}
	\label{tab:pop-stats}
	\begin{tabularx}{0.8\textwidth}{*{6}{Y}}
		\hline Min & $Q_1$ & Median & Mean & $Q_3$ & Max \\
		\hline \numprint{5300000} & \numprint{6775000} & \numprint{9325000} & \numprint{12801000} & \numprint{15675000} & \numprint{70100000} \\ 
		\hline 
	\end{tabularx}
\end{table} 

\subsubsection{Fitting with the initial (unmodified) distributions}

Parameter estimates for the unmodified distributions are reported in Table \ref{tab:pop-results-1} and the different criteria for the unmodified models are reported in Table \ref{tab:pop-comp-1}.

\begin{table}[!h]
	\centering
	\caption{Parameters estimated for unmodified CTP distributions on the Population size dataset}
	\label{tab:pop-results-1}
	\begin{tabularx}{0.8\textwidth}{cYccc}
		\hline Distributions & \multicolumn{4}{c}{Estimations} \\ 
		\hline $CTP_{G}$ & $x_0 = 5300000$ & $\hat\alpha = 1.884$ & $\hat\lambda_1 = 0.534$ & $\hat\lambda_2 = 0.7$ \\ 
		$CTP_{A}$ & $x_0 = 5300000$ & $\hat\alpha = 1.513$ & $\hat\lambda = -0.455$ & \\ 
		$CTP_{R18a}$ & $x_0 = 5300000$ & $\hat\alpha = 0.949$ & $\hat\lambda_1 = 0.003$ & $\hat\lambda_2 = 0.997$ \\ 
		$CTP_{R18b}$ & $x_0 = 5300000$ & $\hat\alpha = 1.884$ & $\hat\lambda_1 = -0.766$ & $\hat\lambda_2 = 0.3$ \\ 
		$CTP_{R19}$ & $x_0 = 5300000$ & $\hat\alpha = 1.295$ & $\hat\lambda = 0.391$ & \\ 
		$CTP_{R23}$ & $x_0 = 5300000$ & $\hat\alpha = 0.969$ & $\hat\lambda = 1$ & $\hat\eta = 1.102$ \\ 
		TP & $x_0 = 5300000$ & $\hat\alpha = 1.787$ & $\hat\lambda = -0.526$ & \\ 
		Pareto & $x_0 = 5300000$ & $\hat\alpha = 1.42$ & & \\ 
		\hline 
	\end{tabularx}
\end{table} 

\begin{table}[!h]
	\centering
	\caption{Criteria for comparing unmodified models on the population size dataset}
	\label{tab:pop-comp-1}
	\begin{tabularx}{0.8\textwidth}{cYYYY}
		\hline Distributions & $-\log L^{*}$ & AIC & AICC & BIC  \\
		\hline  
		$CTP_{R23}$  & $1681.333^{(1)}$ & $3368.666^{(2)}$ & $3368.916^{(3)}$ & $3376.482^{(5)}$ \\ 
		$CTP_{R18a}$ & $1681.370^{(2)}$ & $3368.740^{(4)}$ & $3368.990^{(4)}$ & $3376.556^{(6)}$ \\ 
		$CTP_{G}$    & $1681.422^{(3)}$ & $3368.844^{(5)}$ & $3369.094^{(5)}$ & $3376.660^{(7)}$ \\ 
		$CTP_{R18b}$ & $1681.422^{(3)}$ & $3368.844^{(5)}$ & $3369.094^{(5)}$ & $3376.660^{(7)}$ \\ 
		TP           & $1681.578^{(5)}$ & $3367.156^{(1)}$ & $3367.280^{(1)}$ & $3372.366^{(2)}$ \\ 
		$CTP_{A}$    & $1682.365^{(6)}$ & $3368.730^{(3)}$ & $3368.854^{(2)}$ & $3373.940^{(3)}$ \\ 
		$CTP_{R19}$  & $1682.764^{(7)}$ & $3369.528^{(8)}$ & $3369.652^{(8)}$ & $3374.738^{(4)}$ \\ 
		P            & $1683.638^{(8)}$ & $3369.276^{(7)}$ & $3369.317^{(7)}$ & $3371.881^{(1)}$ \\
		\hline
	\end{tabularx}
\end{table}

According to the results in Table \ref{tab:pop-comp-1}, it seems that the distribution that best fits the population size dataset is $CTP_{R23}$. In this case, the estimated parameters $\hat\lambda_1$ and $\hat\lambda_2$ for the $CTP_G$ are positive so the latter is a probability distribution in the rigorous sense of the term.

\subsubsection{Fitting with modified distributions}

Table \ref{tab:pop-results-2} gives the estimated values of the parameters of the modified models and Table \ref{tab:pop-comp-2} gives the different comparison criteria.

\begin{table}[!h]
	\centering
	\caption{Estimated parameters for modified CTP models on the Population size dataset}
	\label{tab:pop-results-2}
	\begin{tabularx}{0.8\textwidth}{cYccc}
		\hline Distributions & \multicolumn{4}{c}{Estimations} \\ 
		\hline 
		$CTP_{MG}$ & $x_0 = 5300000$ & $\hat\alpha = 0.969$ & $\hat\lambda_1 = 0.898$ & $\hat\lambda_2 = 2.102$ \\ 
		$CTP_{MA}$ & $x_0 = 5300000$ & $\hat\alpha = 1.513$ & $\hat\lambda = -0.455$ & \\ 
		$CTP_{MR18a}$ & $x_0 = 5300000$ & $\hat\alpha = 0.969$ & $\hat\lambda_1 = -0.103$ & $\hat\lambda_2 = 1.102$ \\ 
		$CTP_{MR18b}$ & $x_0 = 5300000$ & $\hat\alpha = 0.969$ & $\hat\lambda_1 = 1$ & $\hat\lambda_2 = -1.102$ \\ 
		$CTP_{MR19}$ & $x_0 = 5300000$ & $\hat\alpha = 1.295$ & $\hat\lambda = 0.391$ & \\ 
		$CTP_{R23}$ & $x_0 = 5300000$ & $\hat\alpha = 0.969$ & $\hat\lambda = 1$ & $\hat\eta = 1.102$ \\ 
		TP & $x_0 = 5300000$ & $\hat\alpha = 1.787$ & $\hat\lambda = -0.526$ & \\ 
		Pareto & $x_0 = 5300000$ & $\hat\alpha = 1.42$ & & \\ 
		\hline 
	\end{tabularx}
\end{table} 

\begin{table}[!h]
	\centering
	\caption{Criteria for modified CTP models on the Population size dataset (ranks in parentheses)}
	\label{tab:pop-comp-2}
	\begin{tabularx}{0.8\textwidth}{cYYYY}
		\hline Distributions & $-\log L^{*}$ & AIC & AICC & BIC  \\
		\hline  
		$CTP_{MG}$    & $1681.333^{(1)}$ & $3368.666^{(2)}$ & $3368.916^{(3)}$ & $3376.482^{(5)}$ \\ 
		$CTP_{MR18a}$ & $1681.333^{(1)}$ & $3368.666^{(2)}$ & $3368.916^{(3)}$ & $3376.482^{(5)}$ \\ 
		$CTP_{MR18b}$ & $1681.333^{(1)}$ & $3368.666^{(2)}$ & $3368.916^{(3)}$ & $3376.482^{(5)}$ \\ 
		$CTP_{R23}$   & $1681.333^{(1)}$ & $3368.666^{(2)}$ & $3368.916^{(3)}$ & $3376.482^{(5)}$ \\ 
		TP            & $1681.578^{(5)}$ & $3367.156^{(1)}$ & $3367.280^{(1)}$ & $3372.366^{(2)}$ \\ 
		$CTP_{MA}$    & $1682.365^{(6)}$ & $3368.730^{(6)}$ & $3368.854^{(2)}$ & $3373.940^{(3)}$ \\ 
		$CTP_{MR19}$  & $1682.764^{(7)}$ & $3369.528^{(8)}$ & $3369.652^{(8)}$ & $3374.738^{(4)}$ \\ 
		Pareto        & $1683.638^{(8)}$ & $3369.276^{(7)}$ & $3369.317^{(7)}$ & $3371.881^{(1)}$ \\ 
		\hline
	\end{tabularx}
\end{table}

A comparison between Tables \ref{tab:pop-results-1} and \ref{tab:pop-results-2} allows us to notice that the parameters of the $CTP_{G}$, $CTP_{R18a}$ and $CTP_{R18b}$ distributions have changed and the distributions which best fit the data are $CTP_{R23}$, $CTP_{MG}$, $CTP_{MR18a}$ and $CTP_{MR18b}$. The corrections made in this paper have improved the $CTP_{G}$, $CTP_{R18a}$ and $CTP_{R18b}$ distributions.

Figure \ref{fig:pop} shows the histogram of Population size data and fitted pdfs of the modified distributions.

\begin{figure}[!h]
	\centering
	\includegraphics[scale=0.5]{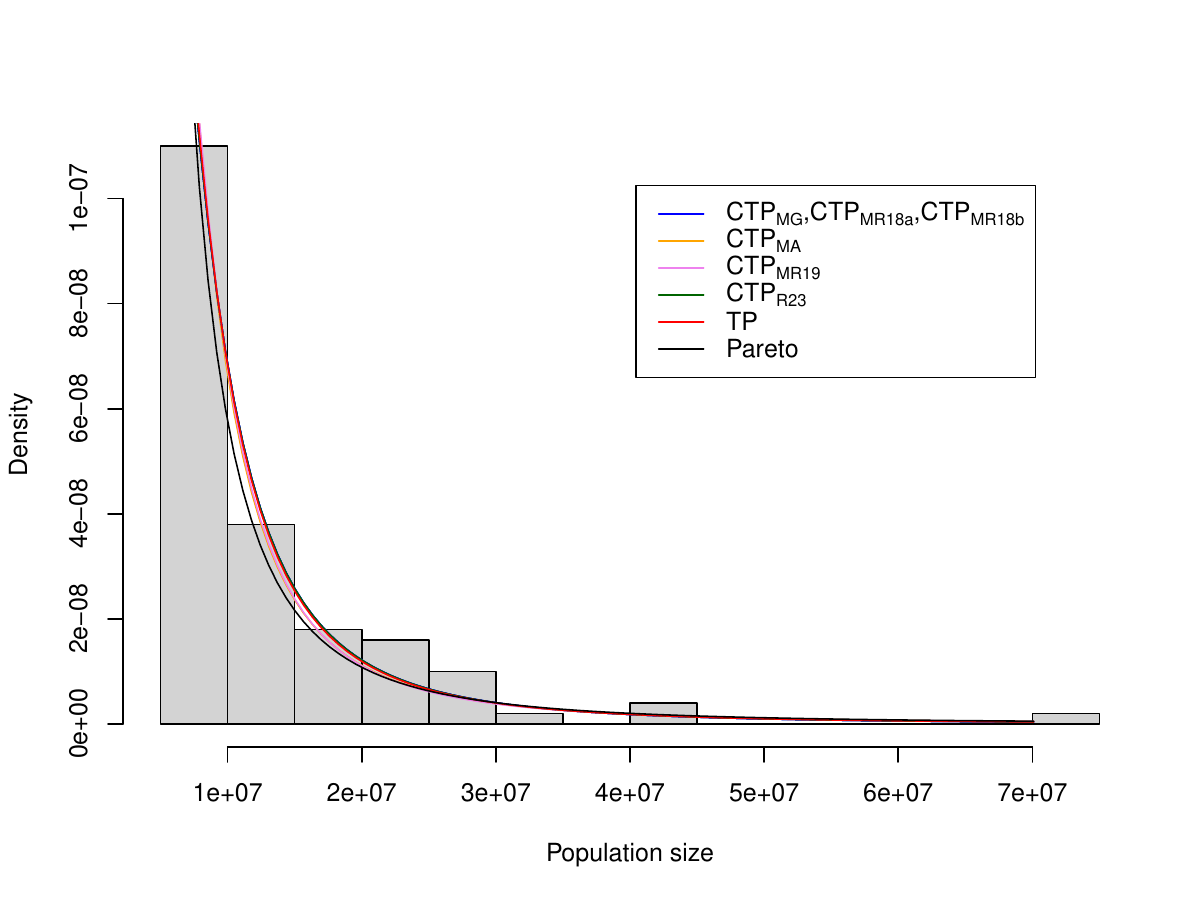}
	\caption{Graphical comparison of the fitting of CTP densities to the Population size dataset}
	\label{fig:pop}
\end{figure}

\section{Discussion and concluding remarks}
\label{sec:conclusion}

The extension of probability distributions is a broad field of research that has gained significant importance over the last two decades, attracting considerable interest from researchers. In this paper, we focused on the transmutation of probability distributions, one of the most widely used techniques for extending probability models.

The first form of transmutation, the quadratic transmutation introduced by \citet{shaw-buckley-2009}, is unique. However, the cubic transmutation developed later is not, leading to multiple formulations in the literature. In this study, we provided a unified definition of cubic transmutation and examined six cubic transmutation formulas proposed by \citet{granzotto-et-al-2017}, \citet{alkadim-mohammed-2017}, \citet{rahman-et-al-2018a}, \citet{rahman-et-al-2018b}, \citet{rahman-et-al-2019b} and \citet{rahman-et-al-2023}. These transmuted cubic distributions involve one or two additional parameters, subject to box and/or linear inequality constraints, in addition to the parameters of the baseline distribution.

Our theoretical and empirical analyses revealed that the cubic transmutation formula of \citet{granzotto-et-al-2017} ($CTP_G$) is not always well defined. Specifically, the constraints imposed on the additional parameters do not always ensure that the resulting cumulative distribution function (cdf) is valid. We demonstrated this issue both theoretically and through empirical illustrations using two real datasets. 

To address these limitations, we proposed modified versions of the models by \citet{granzotto-et-al-2017}, \citet{alkadim-mohammed-2017}, \citet{rahman-et-al-2018a}, \citet{rahman-et-al-2018b} and \citet{rahman-et-al-2019b}, adjusting the parameter ranges to ensure well-defined cdfs. Additionally, we showed that by refining these constraints, the cubic transmutation formulas of \citet{granzotto-et-al-2017}, \citet{rahman-et-al-2018a} and \citet{rahman-et-al-2018b} become equivalent. 

Finally, our real-data applications, using the Pareto distribution as the baseline, confirmed our theoretical findings. The modified models systematically outperformed their unmodified counterparts, reinforcing the necessity of using the improved versions proposed in this paper.

{
	\small 
	\bibliographystyle{apalike} 
	\bibliography{main}
}

\end{document}